\documentclass[11pt]{article}

\usepackage{graphicx}
\usepackage{amsmath}
\usepackage{amssymb}
\usepackage{cite}
\usepackage{booktabs} % table
\usepackage{makecell} % table cell wrapping
\usepackage{enumitem} % nosep
\usepackage{listings}
\usepackage{xcolor}
\usepackage{accsupp}
\usepackage[utf8]{inputenc}
\usepackage{tikz}
\usetikzlibrary{shapes.geometric, arrows.meta, positioning}
\usepackage{pgfplots}
\pgfplotsset{compat=1.18}

\usepackage[margin=1in]{geometry}
\usepackage{authblk}
\usepackage{url}
\usepackage{hyperref}
\usepackage{algorithm}
\usepackage{algpseudocode}
\algblockdefx[EVENT]{Event}{EndEvent}%
  [1]{\textbf{on event}\ \texttt{#1}\ \algorithmicdo}%
  {\algorithmicend\ \textbf{on}}
\algblockdefx[MESSAGE]{Message}{EndMessage}%
  [2][msg]{\textbf{upon message}\ <\texttt{#1}\ |\ #2>\ \algorithmicdo}%
  {\algorithmicend\ \textbf{upon}}
\algnewcommand{\Assert}{\textbf{assert}\xspace}
\MakeRobust{\Call} % for nested Calls

\usepackage{amsthm}
\newtheorem{definition}{Definition}
\newtheorem{theorem}{Theorem}
\newtheorem{proposition}{Proposition}
\newtheorem{lemma}{Lemma}
\newtheorem{corollary}{Corollary}
\newtheorem{remark}{Remark}

\tikzset{
    diagram/.style={font=\small},
    diagram box/.style={rounded corners, draw, align=center,
        minimum width=2.15cm, minimum height=0.75cm},
    diagram arrow/.style={-{Latex}, semithick},
    diagram double arrow/.style={<->, >={Latex}, semithick},
    diagram connector/.style={semithick},
    diagram label/.style={font=\scriptsize, fill=white, inner sep=1.5pt},
}

\title{The Trust-Free Aggregation Layer of the Unicity Infrastructure}

\author[1]{Risto Laanoja}
\author[1]{Mike Gault}
\author[2]{Dirk Draheim}
\author[2]{Ahto Buldas}

\affil[1]{Unicity Labs}
\affil[2]{Tallinn University of Technology}

\date{\today}

\begin{document}
\maketitle

\begin{abstract}
Unicity is a novel blockchain protocol with the ambitious goal of enabling peer-to-peer token transactions to occur off-chain, without shared ordering and execution overhead. This premise requires supporting infrastructure to guarantee that there are no parallel states of assets, or more specifically, that there is no double-spending; a property we term the \textit{unicity}. It turns out that the lack of globally shared state and ordering reduces the blockchain overhead considerably. In designing this infrastructure, no compromises were made regarding its trust assumptions. This paper details the design of the Aggregation Layer, the component responsible for producing Proofs of Inclusion and Non-inclusion to the users. We analyze its design for efficiency and evaluate the robustness of its trust and security model, and design optimal data structures and algorithms for this setup. We then identify the critical property that the Consensus Layer must verify on each round---\emph{append-only consistency}, combining prior-state preservation with coherent placement of insertions---give it a formal definition, and prove that the RSMT consistency proof enforces it over the entire certified history, assuming only collision resistance of the hash function. The security argument is layered: a self-contained single-insertion theory, in which each update is driven by a non-inclusion certificate and the post-state is computed rather than supplied by the untrusted operator, is extended to the batch-optimized transcript, which imports the shape of the touched tree from the operator and pays for its compression with explicitly checked coherent placement. The complete consistency verifier, including coherent key placement, is implemented~\cite{rsmtair} as an Algebraic Intermediate Representation (AIR) circuit on top of the Plonky3~\cite{plonky3} STARK toolkit. The implementation sustains a proving throughput in excess of $10\,000$ insertions per second on a single consumer-class CPU, with a succinct proof and tens of milliseconds verification time, and no trusted setup. Finally, we describe how the per-round proofs of all shards, together with the Consensus Layer's state transitions, are folded into a single recursively aggregated STARK: a fixed-size certificate of the correctness of the system's entire operating history, verifiable against the genesis configuration alone, without trusting the validator set.
\end{abstract}

\section{Motivation}

The foundational principle of the Unicity Network~\cite{wp} is to minimize the volume of on-chain data. This is based on the observation that shared (``on-chain'') state is unavoidable only to prevent double-spending.\footnote{Assuming no centrally controlled, non-transparent technologies such as trusted hardware wallets or Trusted Execution Environments (TEEs); and that anyone can be a recipient} The core tenets of Unicity also include minimizing trust requirements, enhancing user privacy, and providing linear scale.

In a hierarchical trustless system, the principle is that the base layer (e.g., L1 blockchain) provides decentralization, while the layers below it (e.g., rollups) present cryptographic proofs of the correctness of their operation. In scaling Unicity, we have designed efficient data structures to prove the correctness of operation of Aggregation Layer to the Consensus Layer. Based on cryptographic hashes alone, the consistency proof grows linearly with respect to the number of user transactions. This imposes a hard limit of approx. $10\,000$ transactions per second (tx/s), beyond which the networking bandwidth of the Consensus Layer becomes the bottleneck.

To scale further, we use succinct cryptographic proofs to compress the consistency proofs. This use-case is fundamentally more efficient than proving the transaction data itself, as is done in many privacy coins and ZK-rollups: the statement being proved is a single tree update, with the batch as non-public witness data, not the whole execution trace of a virtual machine. The implemented STARK does not claim witness confidentiality.

A useful scoping observation is that the consensus-relevant statement is narrower than ``the SMT was updated correctly''. The Consensus Layer needs to be convinced that \emph{no previously-recorded leaf was deleted or modified} and that \emph{every new leaf is placed coherently with its key}; Remark~\ref{rem:placement} shows, by a concrete attack, that the second half cannot be dropped. Request liveness, submitted-request accountability, and exclusion of unauthorized inserts are self-policed by the protocol layer around the public root commitment: their violation only damages the dishonest aggregator's own ability to serve users, a denial of service by an operationally replaceable component. Section~\ref{sec:scope} details this scoping argument, and Sections~\ref{sec:single-key} and~\ref{sec:consistency-theorem} prove that the resulting statement is sufficient. The narrowness of the in-circuit statement is what makes the AIR small and the proving cheap.

In this paper, we show how to scale the Aggregation Layer to $10\,000$ tx/s \emph{per shard} and beyond. This figure was the original design target; the measured proving throughput of our reference implementation exceeds it by roughly $3\times$ on a single consumer-class CPU (Section~\ref{sec:custom-air-circuit}). Shards process disjoint key ranges and prove their state transitions independently, so aggregate transaction capacity is the sum of their capacities. The BFT Core receives one succinct transition proof per advancing shard rather than one item per transaction; Section~\ref{sec:sharding-architecture} makes this scaling boundary precise.

\section{State of the Art and Comparison}
\label{sec:sota}

The Aggregation Layer occupies a fairly narrow point in the design space of authenticated data structures and succinct arguments. We briefly position it against the most directly comparable lines of work and contrast its proving objectives with those of related systems.

\subsection{Authenticated Append-Only Dictionaries}

Sparse Merkle Trees originate as an extension of binary Merkle hash trees to large key spaces: every potential key is assigned a deterministic position, and the proof that a key is absent is just the root computation taken over canonical-empty subtrees. The systematic study of their cost and engineering trade-offs is more recent. Dahlberg, Pulls, and Peeters~\cite{dahlberg2016smt} present efficient SMT constructions and caching strategies for membership and non-membership proofs with realistic key-space sizes; their analysis is part of the foundation that makes deep-keyed SMTs practical. The path-compressed (radix) variant we use in the Aggregation Layer (Section~\ref{sec:custom-air-circuit}) is a constant-factor refinement of the same line of work.

A related and operationally closer system is Certificate Transparency (CT)~\cite{rfc6962}, which maintains append-only Merkle logs of TLS certificates and provides both inclusion and \emph{consistency} proofs between successive signed tree heads. The cryptographic role of the CT consistency proof is conceptually the same as our consistency proof: it convinces an auditor that a new log state extends the previous one without rewriting it. The differences are operational: CT logs are chronologically append-only with no key-determined position, log entries are not deduplicated against double-spending semantics, and there is no per-round succinct compression of the consistency witness. Our use of SMTs adds a key-determined position (so that double-spending becomes a key collision), while the STARK compresses the per-round witness to a nearly batch-independent proof.

\subsection{Comparison with ZK-Rollups on Ethereum L1}

Ethereum's ZK-rollup ecosystem (Polygon zkEVM~\cite{polygonzkevm}, zkSync, Scroll, StarkNet, and similar systems) uses zero-knowledge proofs to compress the execution trace of an entire L2 virtual machine into a succinct proof that the L1 chain can verify. The statement being proved is, schematically, ``starting from L2 state root $r_{i-1}$ and a block of L2 transactions $B_i$, the L2 EVM produced state root $r_i$''. This requires the circuit to arithmetize the EVM instruction set, account-storage Merkle Patricia tries, signature verification, gas accounting, and so on. The resulting proving cost is large enough that production rollups run dedicated GPU farms or outsource to specialized proving markets.

The Unicity Aggregation Layer proves a much smaller statement: not ``the VM executed correctly'' but only ``the SMT changed in an allowed way''. There is no in-circuit execution, no signature verification, no contract interpretation. Validation of transactions is done off-chain by the parties who have an economic interest in the outcome, namely the recipients in the Execution Layer~\cite{exemodel}. The STARK ensures that the global no-double-spend invariant is preserved. As a result the same cryptographic toolkit (small-field STARKs, FRI, Poseidon-family hashes) delivers $\sim$10$^4$ tx/s per shard on a single CPU rather than $\sim$10$^2$ tx/s on a GPU rack.

\subsection{Comparison with Privacy-Oriented ZK Chains}

Privacy-oriented chains such as Zcash~\cite{zerocash} use ZK-SNARKs to hide transaction contents---sender, receiver, and amount---behind a commitment scheme. A spent note is identified by a public \emph{nullifier} derived from a secret, and the chain maintains a global nullifier set to prevent double-spending. The ZK proof at spend time asserts that the spender knows a commitment in the note tree and the corresponding nullifier without revealing which one. Each user-level transaction therefore carries its own succinct ZK proof.

The Unicity protocol does not use ZK for transaction privacy. Privacy of payloads and recipient information is achieved structurally, by keeping transaction execution off-chain and exposing only an unlinkable per-spend identifier to the global system; the user wallet and service-side privacy aspects of this model are detailed in the companion paper~\cite{exemodel}. A succinct STARK is used for computational integrity, compressing the Aggregation Layer's per-round consistency proof on its way to the Consensus Layer. The prover does not run a circuit per user transaction: one batch proof per round amortizes over thousands of transactions, and there is no anonymity-set scaling cost.

\subsection{Comparison with Nullifier-Tree Mixers}

Nullifier-tree mixers such as Tornado Cash~\cite{tornado} occupy a third design point. They maintain two related authenticated structures: a commitment tree of deposits (a Merkle tree of fixed-denomination notes) and a nullifier set of spent notes. Each withdrawal is accompanied by a ZK proof that the withdrawer knows a deposit commitment in the tree and the associated nullifier, with the nullifier being inserted into the nullifier set to block double-withdrawals. The anonymity set is the population of unspent notes; the privacy guarantee is fundamentally tied to that set's size.

By function, the Unicity Aggregation Layer's SMT is also a structure for preventing double-use of a one-time identifier. The mechanism differs. First, Unicity has no anonymity-set construction: spent state IDs are recorded in clear in the SMT, and unlinkability is provided by the off-chain execution model rather than by hiding the spend inside public blockchain~\cite{exemodel}. Second, the proof boundary and statement are different: in a nullifier-tree mixer a ZK proof is produced \emph{by the user} for each withdrawal and the chain only verifies it, whereas in Unicity a non-ZK STARK is produced \emph{by the aggregator} once per batched round and is not directly shown to users. Third, the underlying tree semantics differ: Tornado's commitment tree is an append-only log used for set-membership proofs, while Unicity's SMT is a keyed dictionary used both for inclusion (Proof of Unicity for the current spend) and non-inclusion (proof that a given state has not yet been spent) proofs.

\subsection{Why STARK Proving in Unicity Is Comparatively Efficient}

The combined effect of the design choices above is a proving workload that is one to three orders of magnitude smaller than in comparable arithmetization-based systems. The principal reasons:

\begin{enumerate}[nosep]
    \item \emph{Narrow in-circuit statement.} The circuit enforces only append-only consistency: prior-state preservation and coherent placement. Unique tree shape follows from those local constraints, while request liveness, accountability, and authorization remain in the surrounding protocol (Section~\ref{sec:scope}).
    \item \emph{No in-circuit transaction execution.} Validation of signatures, predicates, and business logic happens off-chain at the Execution Layer~\cite{exemodel}; the AIR does not see any of it.
    \item \emph{Batch-amortized proving.} One proof per Aggregation Layer round covers thousands of insertions; there is no per-transaction proof and no anonymity-set scaling.
    \item \emph{Local update, sublinear cost in tree capacity.} The proof attests only to the modification of a small subset of the SMT---the paths and siblings touched by the batch---rather than the whole tree. Proving effort scales with the batch size and only logarithmically with the current tree capacity, so a single shard can grow without inflating the per-round prover work.
    \item \emph{Non-public batch witness.} The AIR proves the existence of canonical inserted leaves whose digests reach the public roots, while their contents do not appear in the public statement. This keeps verifier work independent of batch size but, without trace masking, does not by itself provide confidentiality.
    \item \emph{Custom AIR over a small field.} Hand-built arithmetization with an arithmetization-friendly hash (Poseidon2) over BabyBear avoids the 32-bit-to-field translation overhead of general-purpose zkVMs.
    \item \emph{No privacy obligation in-circuit.} Transaction privacy is provided structurally; the STARK is used only for integrity and succinctness, not for zero-knowledge, removing a constraint family that ZK-rollup and mixer designs typically must satisfy.
\end{enumerate}

\section{System Architecture}

To prevent double-spending of tokens, the Unicity Infrastructure permanently\footnote{Permanent from the perspective of a token, meaning for a duration exceeding the token's lifetime.} records a unique identifier for every spent token state. This identifier is the cryptographic hash of the token state data. If a user attempts to double-spend a token, the resulting identifier will be identical to the one already recorded, making it impossible to obtain a new Proof of Unicity. A transaction is considered invalid unless it is accompanied by a valid Proof of Unicity.

The rest of the processing---executing transactions, running smart contracts, etc.---can happen at the client layer, executed by users or ``agents''. Agents are themselves the interested parties in data availability and transaction validation, and they choose the ordering of incoming messages for processing. Thus, the Unicity Infrastructure is relieved of these duties, removing a major scaling bottleneck of traditional L1 blockchains.

The Unicity Infrastructure operates in a trust-minimized way by utilizing distributed authenticated data structures and succinct cryptographic proofs. The Proof of Unicity is a fresh \emph{proof of inclusion} of the token state being spent. This can be efficiently generated based on a Merkle Tree data structure. The proof size is logarithmic with respect to the tree's capacity, making it highly efficient. If the root of the tree is securely fixed, the integrity of the rest of the tree can be verified trustlessly: it is computationally infeasible to generate a valid inclusion proof for an element not present in the tree, without changing the root, or breaking underlying cryptographic assumptions. The infrastructure also supports \textit{non-inclusion proofs}, making it possible to prove to other parties that a particular token state has not yet been spent. The Unicity Infrastructure can thus be conceptualized as a large-scale, distributed Sparse Merkle Tree (SMT). Specifically, the tree is implemented as a path-compressed radix variant, the RSMT (Section~\ref{sec:stack-verifier}), which eliminates single-child internal chains while preserving the standard property that the key uniquely determines the leaf position. Its key-determined layout also gives a canonical horizontal partition: the prefix of an identifier selects exactly one shard, while the remaining bits locate the identifier within that shard (Section~\ref{sec:sharding-architecture}).

Aggregation Layer connects to the Consensus Layer. For fully trustless operation, each request is accompanied by a cryptographic proof of SMT consistency.

\begin{figure}[!htbp]
    \centering
        \begin{tikzpicture}[diagram]
            \node[diagram box, minimum width=4cm, minimum height=1cm]
                (consensus) at (0,0) {Consensus Layer};
            \node[diagram box, minimum width=4cm, minimum height=1cm]
                (aggregation) at (0,-1.5) {Aggregation Layer};
            \node[diagram box, minimum width=4cm, minimum height=1cm]
                (execution) at (0,-3.5) {Execution Layer};
            \draw[diagram connector] (consensus.south) -- (aggregation.north);
            \draw[diagram connector] (aggregation.south) -- (execution.north);
            \draw[diagram connector, dashed] (-1,-2.5) -- (1,-2.5);
        \end{tikzpicture}
    \caption{Layered architecture of the Unicity Network.}\label{fig:layers}
\end{figure}
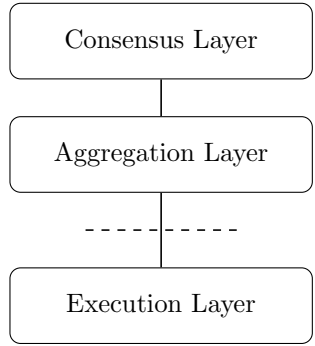

\subsection{Consensus Layer}
\label{sec:consensus-layer}

The Consensus Layer consists of a single logical component, the \emph{BFT Core}: a bounded-size committee of validator nodes running a Byzantine fault tolerant (BFT) consensus protocol. During a round of execution, it receives certification requests from Aggregation Layer shards, checks that each request extends the previously certified state of its shard, verifies the accompanying consistency proof, and issues a \emph{Unicity Certificate} over the updated global state. BFT consensus provides deterministic finality; its usual quorum assumption is needed for prompt progress and for uniqueness of the most recent certified tip.

The BFT Core maintains no blockchain. Its persistent protocol state is cumulative: the vector of most recently certified shard roots (combined into a single global root, Section~\ref{sec:data-flow}) and the \emph{Unicity Trust Base}, an authenticated record of the validator set and quorum rule applicable to each configuration period. There are no transaction blocks: ordering and availability of user requests are handled below, at the Aggregation and Execution Layers, and issued certificates are persisted by the parties who need them and by the public round archive (Section~\ref{sec:data-availability}).

Committee formation and its operational policy are orthogonal to the Aggregation Layer and outside the scope of this paper. For the immediate validation path we assume an authenticated Trust Base and a non-equivocating BFT quorum. Section~\ref{sec:aggregation-audit} then shows how recursive proof aggregation removes the quorum from the correctness argument for the recorded history; equivocation between otherwise valid tips remains detectable from conflicting signed artifacts.

\subsection{Aggregation Layer}
\label{sec:sharding-architecture}

The Aggregation Layer implements a global, append-only key-value store that immutably records every spent token state. More specifically, it provides the following services: 1) recording of key-value tuples where the key identifies a token state and value is recording some meta-data, 2) returning inclusion proofs of keys, 3) returning non-inclusion proofs of keys not present in the store.

The Aggregation Layer periodically has its state authenticator certified by the Consensus Layer.

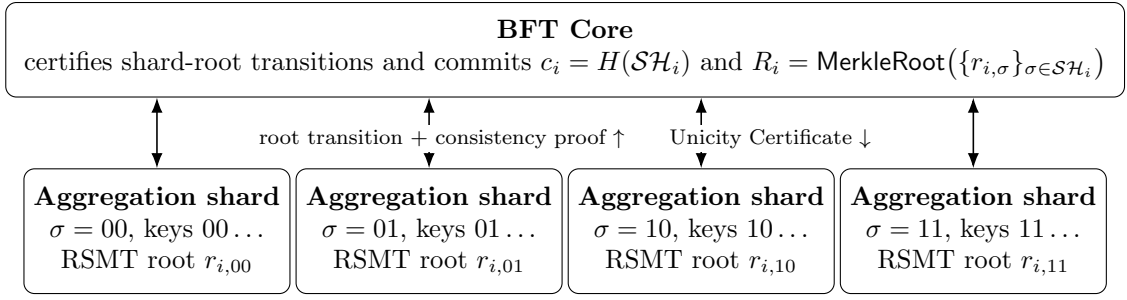
\begin{figure*}[!t]
    \centering
    \begin{tikzpicture}[
        diagram,
        corebox/.style={diagram box, minimum width=14.8cm, minimum height=1.25cm},
        shardbox/.style={diagram box, minimum width=3.1cm, minimum height=1.65cm},
        certflow/.style={diagram double arrow}
    ]
        \node[corebox] (core) at (0,2.0) {\textbf{BFT Core}\\
            certifies shard-root transitions and commits $c_i=H(\mathcal{SH}_i)$ and
            $R_i=\mathsf{MerkleRoot}\bigl(\{r_{i,\sigma}\}_{\sigma\in\mathcal{SH}_i}\bigr)$};

        \node[shardbox] (s00) at (-5.4,-0.4) {\textbf{Aggregation shard}\\
            $\sigma=00$, keys $00\ldots$\\RSMT root $r_{i,00}$};
        \node[shardbox] (s01) at (-1.8,-0.4) {\textbf{Aggregation shard}\\
            $\sigma=01$, keys $01\ldots$\\RSMT root $r_{i,01}$};
        \node[shardbox] (s10) at (1.8,-0.4) {\textbf{Aggregation shard}\\
            $\sigma=10$, keys $10\ldots$\\RSMT root $r_{i,10}$};
        \node[shardbox] (s11) at (5.4,-0.4) {\textbf{Aggregation shard}\\
            $\sigma=11$, keys $11\ldots$\\RSMT root $r_{i,11}$};

        \draw[certflow] (s00.north) -- (s00.north |- core.south);
        \draw[certflow] (s01.north) -- (s01.north |- core.south);
        \draw[certflow] (s10.north) -- (s10.north |- core.south);
        \draw[certflow] (s11.north) -- (s11.north |- core.south);

        \node[diagram label, align=center]
            at (0,0.82) {root transition + consistency proof $\uparrow$\qquad
                         Unicity Certificate $\downarrow$};
    \end{tikzpicture}
    \caption{Aggregation shards beneath one logical BFT Core. The four equal prefixes are illustrative.}\label{fig:sharding}
\end{figure*}

\paragraph{Deterministic keyspace partition.}
Let the keys be $\kappa$-bit strings. A sharding scheme $\mathcal{SH}\subseteq\{0,1\}^{*}$ is a code such that for every $k\in\{0,1\}^\kappa$ there is a unique $w\in\mathcal{SH}$ (denoted by $f_{\mathcal{SH}}(k)$) such that $w\preceq k$, i.e.
every key has a unique prefix in $\mathcal{SH}$. The codewords of $\mathcal{SH}$ are used as shard identifiers and the function
$f_{\mathcal{SH}}\colon \{0,1\}^\kappa\rightarrow \mathcal{SH}$ is called the routing function.
Every shard $\sigma\in\mathcal{SH}$ maintains an independent RSMT $T_\sigma$ containing exactly the bindings whose keys satisfy $f_{\mathcal{SH}}(k)=\sigma$. Routing therefore depends only on public information: the bits of the key. No directory lookups are needed. Requests related to keys of different shards can be batched, inserted, and proved concurrently.

\paragraph{Authentication across shards.}
Let $\mathcal{SH}_i$ denote the sharding scheme used in round $i$. The latest root hashes $r_{i,\sigma}$ of all shards $\sigma\in\mathcal{SH}_i$ are the leaves of a Merkle \emph{shard-root tree}; a shard $\sigma$ that does not advance retains its previous root hash $r_{i-1,\sigma}$. The root of the shard root tree is a pair $(c_i,R_i)$, where $R_i$ is the root hash of the entire Aggregation Layer state, while $c_i=H(\mathcal{SH}_i)$ binds to the sharding scheme. A Unicity Certificate for shard $\sigma$ authenticates $c_i$, authenticates $r_{i,\sigma}$ to $R_i$ with the sibling hashes on the prefix path, and authenticates $R_i$ with the BFT Core's quorum certificate. Consequently, an inclusion or non-inclusion proof has two independent parts: a local RSMT path within $T_\sigma$, and a short shard-root path to $R_i$. Neither part grows with transaction throughput; the latter contains $|\sigma|$ sibling hashes.

\paragraph{Dynamic shard splitting.}
A busy shard can be split by replacing one prefix $\sigma$ in $\mathcal{SH}$ with the two children $\sigma\|0$ and $\sigma\|1$. The children retain the parent leaves selected by the next key bit. Because RSMT leaves commit to full keys and internal nodes commit to absolute bifurcation depths and key regions (Section~\ref{sec:stack-verifier}), each child root is either the hash of an existing parent subtree or the canonical empty root; retained nodes need not be rehashed and the certified history need not be replayed. The split is activated as a configuration change, after which both children evolve and prove independently. A local RSMT certificate may lose the junction at the split depth while its shard-root certificate gains one sibling hash, moving authentication work from the shard-local tree to the common root without introducing another aggregation tier.

\paragraph{Horizontal capacity and proof aggregation.}
If shard $\sigma$ sustains insertion rate $q_\sigma$, the Aggregation Layer sustains
\[
    Q=\sum_{\sigma\in\mathcal{SH}} q_\sigma
\]
subject to the BFT Core's capacity for shard summaries. A shard round containing thousands of insertions exports only $(r_{i-1,\sigma},r_{i,\sigma},\pi_{i,\sigma})$, where the succinct consistency proof $\pi_{i,\sigma}$ has verifier cost independent of the number of stored leaves and is verified independently of other shards. Thus, adding a shard adds storage, batching, and proving capacity without increasing any existing shard's workload; the Core's work grows with the number of advancing shards, not with $Q$. These verifications are mutually independent and can be parallelized. The off-critical-path construction of Section~\ref{sec:aggregation-audit} subsequently folds all changed-shard proofs and the corresponding $R_i$ transitions into one recursively updated proof whose public statement remains fixed-size regardless of the number of shards or elapsed rounds.

\paragraph{Per-shard consistency.}
Once a key is set, it must remain there permanently. Every shard transition is therefore accompanied by a cryptographic proof that pre-existing keys were neither removed nor modified and that new keys were placed at the positions determined by the keys themselves. The direct hash-based witness grows with the insertion batch, whereas the STARK construction of Section~\ref{sec:custom-air-circuit} makes the public statement just about the old and new shard roots and gives succinct verification independent of the batch size. Correct verification and BFT chaining of those roots make each shard an untrusted, cryptographically checked service.

\subsection{Execution Layer}

The Execution Layer is responsible for executing transactions and other business logic, using the services of the Aggregation Layer and Unicity in general. Its formal security model---including double-spending resistance, non-blocking, and service- and user-side privacy---is developed in~\cite{exemodel}. Programmable ownership predicates extend this model with off-chain smart-contract functionality~\cite{predicates}.

\subsection{Data Flow}
\label{sec:data-flow}

Figure~\ref{fig:dataflow} summarizes the flow of authenticated data through the system, from the shards to the auditing verifier. In BFT Core round $i$, each shard $\sigma$ that advances submits a certification request carrying its previous root $r_{i-1,\sigma}$, its new root $r_{i,\sigma}$, and the consistency proof $\pi_{i,\sigma}$; roots of non-advancing shards are carried forward. The BFT Core verifies the submitted proofs, combines the current shard roots into the global root $R_i$, reaches consensus, and returns Unicity Certificates to the advancing shards. Since the BFT Core keeps no blockchain, the per-round artifacts---shard roots, consistency proofs, certificates, and Trust Base entries---are published to the public round archive (Section~\ref{sec:data-availability}). From the archive, an aggregation prover folds the history into a single constant-size proof $\Pi_n$ (Section~\ref{sec:aggregation-audit}), which any party can verify against the genesis configuration.

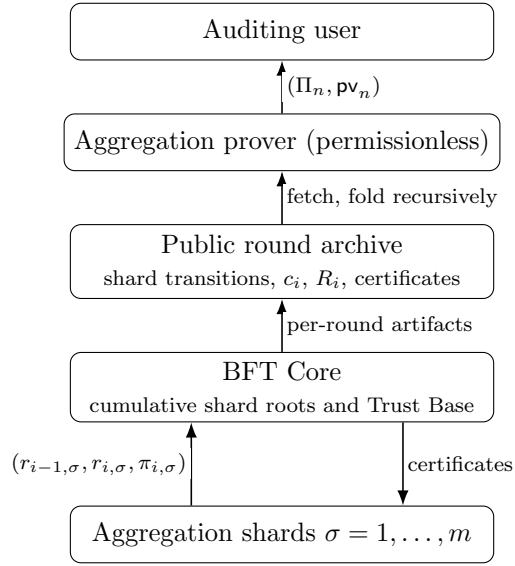
\begin{figure}[!htbp]
    \centering
\begin{tikzpicture}[diagram,
  dfbox/.style={diagram box, minimum width=5.6cm}]
  \node[dfbox] (verifier) {Auditing user};
  \node[dfbox, below=0.7cm of verifier] (prover) {Aggregation prover (permissionless)};
  \node[dfbox, below=0.7cm of prover] (arch) {Public round archive\\ {\scriptsize shard transitions, $c_i$, $R_i$, certificates}};
  \node[dfbox, below=0.7cm of arch] (core) {BFT Core\\ {\scriptsize cumulative shard roots and Trust Base}};
  \node[dfbox, below=1.1cm of core] (shards) {Aggregation shards $\sigma = 1, \ldots, m$};
  \draw[diagram arrow] ([xshift=-1.2cm]shards.north) -- node[diagram label, left] {$(r_{i-1,\sigma}, r_{i,\sigma}, \pi_{i,\sigma})$} ([xshift=-1.2cm]core.south);
  \draw[diagram arrow] ([xshift=1.6cm]core.south) -- node[diagram label, right] {certificates} ([xshift=1.6cm]shards.north);
  \draw[diagram arrow] (core) -- node[diagram label, right] {per-round artifacts} (arch);
  \draw[diagram arrow] (arch) -- node[diagram label, right] {fetch, fold recursively} (prover);
  \draw[diagram arrow] (prover) -- node[diagram label, right] {$(\Pi_n, \mathsf{pv}_n)$} (verifier);
\end{tikzpicture}
    \caption{Data flow from the shards to the auditing verifier. The pragmatic validation path (Section~\ref{sec:practical}) uses the certificates directly; the audit path uses the aggregate proof $\Pi_n$.}\label{fig:dataflow}
\end{figure}

\section{Security Model of the Aggregation Layer}

The Aggregation Layer implements a distributed, authenticated, append-only dictionary data structure. It authenticates incoming state transfer certification requests (store new key-value pairs $(k,v)$ to the accumulator) by verifying that the sender possesses the private key corresponding to the public key that identifies the current token owner. The specific authentication protocol is beyond the scope of this paper.

After every round $i$, the state of the aggregator represents a partial function $M_i$ such that for every key $k\in\{0,1\}^\kappa$,
either $M_i(k)=v$ (if a pair $(k,v)$ was accumulated) or $M_i(k)=\bot$ (if no pairs $(k,v)$ were accumulated). As $M_i$ is required to be a partial function, from $M_i(k)=v\neq\bot$ it follows that no other pairs $(k,v')$ with $v'\neq v$ can be accumulated. Initially, the partial function is $M_0=\emptyset$, i.e. $M_0(k)=\bot$ for every key $k$.

It is also required that for every $i$, the function $M_{i+1}$ extends the function $M_i$ (denoted by $M_i\subseteq M_{i+1}$), i.e. if $M_i(k)=v\neq\bot$, then $M_{i+1}(k)=v$. This property formalizes the append-only property of the aggregator.

The partial function $M_i$ is certified by computing a digest $r_i=\mathsf{Dig}(M_i)$, which is implemented as the root hash $\mathsf{dig}(T)$ of the RSMT $T=\mathsf{Tree}(M_i)$ of the set of all pairs $(k,v)\in M_i$. Two kinds of cryptographic proofs are used:
\begin{itemize}
\item \emph{Inclusion/Non-inclusion proofs} that certify statements $M_i(k)=v$ and are verified against the digest $r_i$, given the pair $(k,v)$. If $v\neq\bot$, the proof is called an inclusion proof, otherwise it is called a non-inclusion proof. These proofs are used by users to check the validity of token transfers.
\item \emph{Consistency proofs} that certify the statements $M_i\subseteq M_{i+1}$ and are verified against the digests $r_i$ and $r_{i+1}$. These proofs are used by the Consensus Layer to check that the previous contents of the aggregator were not modified during the round $i+1$.
\end{itemize}

If $M_i(k)=v$, then the inclusion/non-inclusion proof is easy to construct and if $M_{i}\subseteq M_{i+1}$ then the consistency proof is easy to construct.

We require the following security properties:

\begin{definition}[Accumulator interface]
\begin{enumerate}[nosep,label=(\roman*)]
  \item From any consistency proof that verifies against $(r_i,r_{i+1})$, partial functions $M_i\subseteq M_{i+1}$ with $\mathsf{Dig}(M_i)=r_i$ and $\mathsf{Dig}(M_{i+1})=r_{i+1}$ are efficiently computable, or a collision of the underlying hash function is.
  \item From $(k,r,v,v')$ with $v\neq v'$ and two inclusion/non-inclusion proofs, one for $M(k)=v$ and one for $M(k)=v'$, that both verify against the same digest $r$, a collision of the underlying hash function is efficiently computable.
\end{enumerate}
\label{def:append-only-accumulator}
\end{definition}

After each batch of additions, the new root of the Aggregation Layer's SMT is certified by the BFT Core, ensuring its uniqueness and immutability. This provides a secure trust anchor for all consistency, inclusion, and non-inclusion proofs. The idealized Consensus Layer is modeled as Algorithm~\ref{alg:consensuslayer}.

\begin{figure}[!htbp]
    \centering
\begin{tikzpicture}[diagram, node distance=2cm and 2cm]

    % Nodes
    \node[diagram box, minimum width=5cm] (consensus) {Consensus Layer};
    \node[diagram box, minimum width=3cm, below=of consensus] (smt) {Aggregator (RSMT)};
    \node[diagram box, minimum width=5cm, below=of smt] (users) {Token Users};

    % Arrows
    \draw[diagram arrow] ([xshift=-0.1cm]smt.north) -- ++(0,1) node[diagram label, left] {$(r_i, r_{i-1}, \pi)$} -- ([xshift=-0.1cm]consensus.south);
    \draw[diagram arrow] ([xshift=0.1cm]consensus.south) -- node[diagram label, right] {$c = (i, r_i, r_{i-1}; s_{\textsf{cl}})$} ([xshift=0.1cm]smt.north);
    \draw[diagram arrow] ([xshift=-0.4cm]users.north) -- node[diagram label, left] {$B_i = ((k_1,v_1), \ldots, (k_j,v_j))$} ([xshift=-0.4cm]smt.south);
    \draw[diagram arrow] ([xshift=0.4cm]smt.south) -- ++(0,-1) node[diagram label, right] {
            $\substack{
                \text{inclusion / non-inclusion certificates}\\
                \text{against certified } r_i
            }$
            } -- ([xshift=0.4cm]users.north);

\end{tikzpicture}
    \caption{Security model of the Aggregation Layer.}\label{fig:model}
\end{figure}
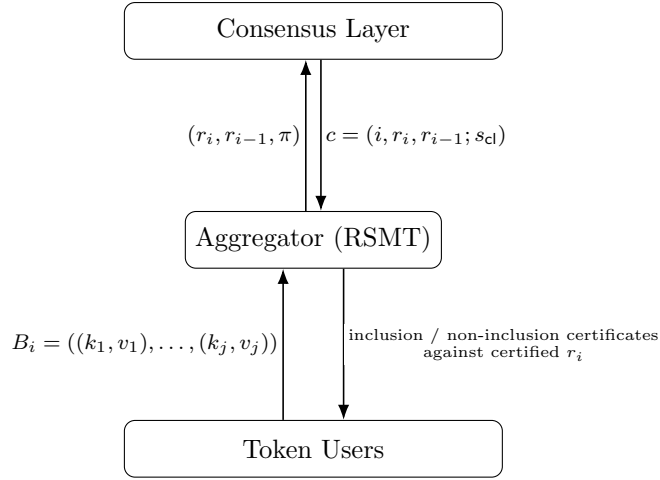

For efficiency reasons client requests are processed in batches; the tree is re-calculated and the tree root is certified when a batch is closed. A batch of client requests is denoted as $B_i$. At the end of each batch, the Aggregation Layer produces its summary root hash $r_i$ and sends it to the Consensus Layer for certification. A certification request $(r_i, r_{i-1}, \pi)$ includes: 1) the new state root hash, 2) the previous state root hash, and 3) a consistency proof of the changes made during the batch.

The Consensus Layer certifies the request only if it uniquely \textit{extends} a previously certified state root and the consistency proof is valid. It returns a certificate $c = (i, r_i, r_{i-1}; s_{\textsf{cl}})$, where $s_{\textsf{cl}}$ is a signature from the Consensus Layer (e.g., a threshold signature from the consensus nodes or a proof of inclusion in a finalized block).

Each state can be extended only once, which prevents forks within the Aggregation Layer. Each subsequent round extends the most recently certified state.
We model the Consensus Layer as an oracle, as shown in Algorithm~\ref{alg:consensuslayer}.

\begin{algorithm}[tbh]
  \caption{Consensus Layer modeled as an oracle}\label{alg:consensuslayer}
  \begin{algorithmic}[0]
    \Function{Initialize}{\null}
        \State $r_- \gets \bot$
        \State $i \gets 0$
    \EndFunction
    \Function{CertificationRequest}{$r_i, r_{i-1}, \pi$}
        \If {$(r_{i-1} \ne  r_-) \lor \lnot\Call{valid}{\pi, r_i, r_{i-1}}$}
            \State \Return $\bot$
        \EndIf
        \State $r_- \gets r_i$
        \State $i \gets i+1$
        \State $s_{\textsf{cl}} \gets \textsf{sig}_\textsf{cl}(i, r_i, r_{i-1})$
        \State \Return $c = (i, r_i, r_{i-1}; s_{\textsf{cl}})$
    \EndFunction
  \end{algorithmic}
\end{algorithm}

The SMT provides users with inclusion and non-inclusion proofs. Each proof is anchored to a state root certified by the Consensus Layer.

The Consensus Layer must guarantee data availability. If recent state roots were lost, it would become impossible to reject duplicate state transition requests, potentially allowing malicious actors to double-spend against an old, un-extendable state. The Aggregation Layer itself does not require an internal consensus mechanism; protocols like Raft could be used for replication and coordination among its redundant nodes. The decentralized consensus is provided by the external Consensus Layer.

If each state transition is accompanied by a cryptographic consistency proof (see Section~\ref{sec:consistency-proof}), the Aggregation Layer can be considered trustless.

\subsection{``Maximalist'' Security Assumptions}
\label{sec:maximalist}

In this model, we assume that users wish to validate all aspects of system operation that are relevant to their own assets, accepting no assumptions beyond standard cryptographic ones. This level of trustlessness is close to the strong guarantees introduced by Bitcoin~\cite{bitcoin}, where each ``client'' functions as a full validator, starting from downloading and verifying the blockchain from the genesis block. In Unicity there is no blockchain to replay; the equivalent guarantee is obtained from a single succinct proof.

Upon receiving a token, the user must be able to efficiently verify the following:

\begin{enumerate}[nosep]
    \item The token is valid (as elaborated elsewhere),
    \item The Aggregation Layer has not forked,
    \item The Aggregation Layer has not certified conflicting states of the same token.
\end{enumerate}

The second and third points are covered by the aggregated history proof of Section~\ref{sec:aggregation-audit}: a fixed-size STARK, updated periodically, attesting that the entire sequence of certified states---from the genesis configuration up to a recent round $n$---forms a single non-forking chain in which every round of every shard satisfies append-only consistency (Definition~\ref{def:aoc}, established by Theorem~\ref{thm:history}). In particular, no inclusion proof can exist for a token state that is absent from the recorded history. Verifying the aggregate proof requires only the genesis parameters of the network instance and takes milliseconds on commodity hardware; neither replay of history nor trust in the validator set is involved.

The aggregate proof is produced with a latency of minutes to hours behind the certified tip, so maximalist verification is not instantaneous. It is best understood as an audit mechanism: it retrospectively confirms---or refutes, with the failure round pinpointed---the correctness of operation of the Consensus and Aggregation Layers. The validation procedure is given in Section~\ref{sec:maxi-validation}.

\subsection{Practical Security Assumptions}
\label{sec:practical}

For immediate finality, we assume that a BFT quorum follows the protocol and does not collude maliciously with the Aggregation Layer. Under this standard consensus assumption, users obtain substantially better latency than the delayed full-history audit. BFT layer forking (case 2 above) or certification of conflicting states (case 3 above) produces strong cryptographic evidence that can be processed out of the critical path of serving users.

In this scenario, a transaction is finalized, and an inclusion proof is returned within a few seconds, allowing the transaction to be independently verified—without consulting external data\footnote{Previously obtained Root of Trust is used to validate future transactions}—within the same timeframe.

The Root of Trust is the Unicity Trust Base: the chain of epoch records of the BFT Core (Section~\ref{sec:consensus-layer}). These records grow slowly---one aggregated-signature record per validator-set change---and validating a certificate requires checking its signatures against the applicable epoch record only. This validation path is available immediately, within the round time; the delayed audit path of the maximalist model complements rather than replaces it.

\subsection{Scope of the In-Circuit Statement}
\label{sec:scope}

A core engineering choice in the design of the Aggregation Layer's consistency proof is to push as much of the per-round security argument out of the cryptographic proof as possible, leaving only a small kernel inside the circuit. The value of the scoping is engineering economy---it keeps the circuit small. Drawing the boundary correctly, however, requires care: it is tempting to place \emph{all} placement-related properties outside the kernel, on the argument that a misplaced insertion only damages the aggregator's own ability to serve proofs. Remark~\ref{rem:placement} refutes that argument by a concrete attack; the boundary is drawn below and proved sufficient in Sections~\ref{sec:single-key} and~\ref{sec:consistency-theorem}.

Let $M_{i-1}, M_i \colon \{0,1\}^* \to \{0,1\}^* \cup \{\bot\}$ be the partial maps of recorded key-value bindings committed by, respectively, the previous and the new state roots $r_{i-1}, r_i$ (Section~\ref{sec:consistency-formal} makes ``committed by'' precise). We say the round update $r_{i-1} \to r_i$ satisfies \emph{prior-state preservation} iff
\begin{equation}
  \forall k \in \mathrm{dom}(M_{i-1})\colon\quad M_i(k) = M_{i-1}(k).
  \label{eq:psp}
\end{equation}
That is, every key already recorded under $r_{i-1}$ is bound to the same value under $r_i$. Equivalently, the round adds a (possibly empty) set of fresh keys and modifies nothing.

\begin{definition}[In-circuit statement, informal]
\label{def:minimal}
For the Aggregation Layer, the cryptographic per-round consistency proof must enforce, given authentic roots $(r_{i-1}, r_i)$: (i) prior-state preservation \eqref{eq:psp}, and (ii) \emph{coherent placement}: every leaf inserted in the round sits at the position determined by its own key. Definition~\ref{def:aoc} states this formally.
\end{definition}

Clause (ii) is important. Because the maps $M_i$ are realized by \emph{provability} against the root---a binding is ``recorded'' exactly when an inclusion proof for it verifies---an incoherently placed insertion changes which bindings are provable, for old keys as well as new ones. Remark~\ref{rem:placement} shows that dropping (ii) admits verifying round transitions under which an already-recorded key becomes re-recordable with a different value: the equivocation that \eqref{eq:psp} is meant to exclude.

The remaining desirable properties of the Aggregation Layer stay outside the kernel, and can be enumerated and accounted for as follows:

\begin{description}[nosep]
    \item[Request liveness] (\emph{every well-formed user request is eventually recorded}). A censoring aggregator does not violate the in-circuit statement; it merely fails to serve some users. The protocol mitigates this by replication (highly-available cluster) and by allowing users to resubmit through alternative aggregators in the same shard.
    \item[Submitted-request accountability] (\emph{omission of an accepted request is detectable}). The consistency proof binds every element of its witness batch to the new root, but it does not prove that this non-public witness equals the external queue of authenticated requests. A recipient detects omission by requesting an inclusion certificate against the certified root.
    \item[No phantom inserts] (\emph{nothing is recorded that was not a user request}). At worst, an aggregator does free recording work for itself or third parties; phantom entries carry fresh keys (coherent placement forbids re-recording), so they cannot affect any honest user's tokens.
    \item[Unique Patricia shape] (\emph{the post-state tree is the unique tree prescribed by the data structure's rules for its key set}). This does not need to be postulated as a separate in-circuit obligation: the local validity conditions already force that shape (Lemma~\ref{lem:shape} and Proposition~\ref{prop:unique}), and Theorem~\ref{thm:history} preserves those conditions from the empty genesis tree.
    \item[Append-only consistency] (\emph{\eqref{eq:psp} together with coherent placement; Definition~\ref{def:aoc}}). \emph{Critical.} If this fails, the operator can rewrite history and double-spending becomes possible. This is the property that must be cryptographically enforced \emph{before} round certification.
\end{description}

Restricting the in-circuit statement to append-only consistency has several useful consequences. First, the statement need only assert the existence of canonical inserted leaves; the batch contents are non-public witness data rather than public inputs (Section~\ref{sec:custom-air-circuit}). Semantically the public transition is the old and new root pair, augmented in the proof envelope by an old-root absence flag, a protocol identifier, and scalar trace-shape metadata. Second, the circuit does not have to \emph{reconstruct} the unique global tree shape from scratch, which is by far the most constraint-heavy aspect of any naive arithmetization (cf.~\cite{rsmtair}, ``cost-of-canonical detour''); with the region-committing hash of Section~\ref{sec:stack-verifier}, coherent placement is checked by local constraints on the touched nodes only. Third, the proof scales with the batch size alone, not with the total tree capacity.

\section{Design Rationale for the Authenticated Dictionary}
\label{sec:tree-design}

The authenticated dictionary lies on two critical paths. Users obtain inclusion and non-inclusion certificates and store them in the tokens, in many copies, so their size should be optimal --- the limit is one sibling digest per actual branch. The Consensus Layer verifies a consistency proof for every batch and for every shard; the proof should be short, the algorithm efficient and inexpensive to arithmetize. At the same time, the operator is untrusted, so neither compactness nor speed may weaken append-only consistency. These requirements lead to the design below.

\subsection{Indexed and Radix Trees}

A conventional indexed SMT (\emph{Sparse Merkle Tree}) assigns every key to a leaf of a fixed-depth binary tree. Empty subtrees are implicit, and storage and proofs can omit single-child chains. Its main security advantage is that the placement is part of the tree semantics, so a prover cannot choose an alternative path for a key. A batch-consistency witness supplies the untouched sibling subtrees around the batch, and the verifier reconstructs the old root with the batch positions empty and the new root with them populated. A post-order stack encoding makes this construction reasonably regular, although stack entries and branch instructions must still carry absolute positions in the indexed tree.\footnote{Experiment: \href{https://github.com/ristik/ndsmt-experiments/blob/main/ndsmt_op.py}{https://github.com/.../ndsmt\_op.py}}

An alternative is breadth-first traversal, where the witness is a per-layer array of ordered missing siblings necessary to reconstruct the old and new roots. This results in a one-pass (per root) linear verification algorithm, without lookups.\footnote{Experiment: \href{https://github.com/ristik/ndsmt-experiments/blob/main/ndsmt_lvl.py}{https://github.com/.../ndsmt\_lvl.py}}

A Patricia or radix tree makes path compression part of the committed structure instead. It stores only leaves and bifurcations, rather than making an intractably large tree manageable through optimizations. Hashing the full key with the value anchors each leaf, and committing an internal node to its absolute bifurcation depth makes every existing node immutable when a later insertion splits an edge above it. The touched post-state can then be serialized with three basic instructions: a preserved subtree \(S\), a new leaf \(L\), and a junction \(N\). A stack machine evaluates the same transcript twice, passing preserved hashes through on the old side and constructing the new root on the other. This is a natural shape for tabular AIR arithmetization: the verifier is a single forward scan (passes to compute pre-state and post-state are natural to combine) with bounded stack state and one local rule per instruction.

Path compression also gives compact user proofs. A certificate needs one sibling hash per junction on the key-directed path, a bitmap (or list) of their depths, and the leaf at which the path ends; on random keys the number of such junctions grows logarithmically with the capacity. Because key-directed descent is deterministic, the same format serves non-inclusion: the path simply ends at a different leaf. But, the extra complexity is that the canonical data-structure construction must be validated by the verifier, specifically the proof must establish that each prover-supplied junction is in the correct key-space region.

\subsection{Authenticating Placement}

Committing a radix junction only to its two children and bifurcation depth gives the smallest structural verifier,\footnote{Experiment: \href{https://github.com/ristik/ndsmt-experiments/blob/main/ndrsmt3o.py}{https://github.com/.../ndrsmt3o.py}} but it is insufficient for trustless operation. As Remark~\ref{rem:placement} demonstrates, an operator can attach a preserved tree to the wrong side of a new junction and insert a second value for an already recorded key on the key-directed side. Both root computations still agree with the transcript. Thus, preservation of hashes is not yet preservation of the authenticated dictionary: the new edges must also be coherent with their keys.

There are two natural ways to authenticate this missing information. A \emph{range commitment} includes the smallest and largest key below every junction. It lets the verifier authenticate a split using the adjacent child boundaries, but the cost reaches user proofs: each sibling must carry two additional 256-bit keys as well as its digest.\footnote{Experiment: \href{https://github.com/ristik/ndsmt-experiments/blob/main/rsmt4.py}{https://github.com/.../rsmt4.py}}
A \emph{region commitment} instead includes the junction's absolute key prefix \(p\) alongside its depth \(d\). Whenever an insertion creates an edge, the consistency verifier checks locally that the child is deeper and that its region extends \(p\|0\) or \(p\|1\), according to its side. Old--old edges need not be reopened because the hash freezes checks made in earlier rounds.\footnote{Experiment: \href{https://github.com/ristik/ndsmt-experiments/blob/main/rsmt6.py}{https://github.com/.../rsmt6.py}}

The region is derivable from the terminal leaf's key during certificate verification, so it adds no bytes to either inclusion or non-inclusion. It need not appear in every \(N\) instruction either: the consistency verifier derives the parent region from authenticated child advice and requires both children to be opened where a new junction meets a preserved subtree. The resulting proof retains the small post-order machine while closing the placement attack.\footnote{Experiment: \href{https://github.com/ristik/ndsmt-experiments/blob/main/rsmt6a.py}{https://github.com/.../rsmt6a.py}} Formally, the region commitment is load-bearing only for the transcript-based consistency proof: certificate verification and the single-insertion update protocol never check it (Remark~\ref{rem:derived-vs-checked}).

\begin{figure*}[t]
    \centering
    \includegraphics[width=0.92\textwidth]{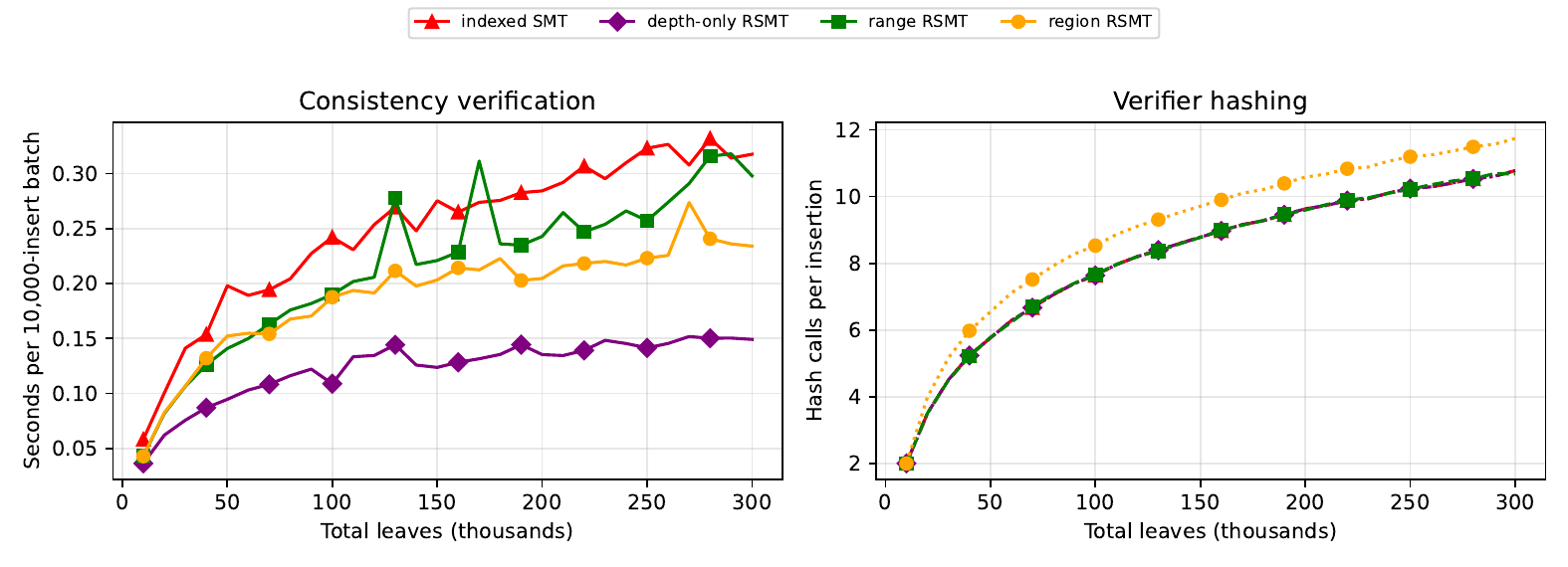}
    \caption{Native consistency-proof costs for the principal design points. Each round inserts \(10\,000\) uniformly distributed keys; 30 rounds grow the tree to \(300\,000\) leaves. The left panel reports the median of three Python/SHA-256 verification runs and the right panel counts native hash invocations. The depth-only RSMT is included as a performance lower bound but does not enforce coherent placement.}
    \label{fig:tree-design-consistency}
\end{figure*}

\begin{figure}[t]
    \centering
    \includegraphics[width=.5 \textwidth]{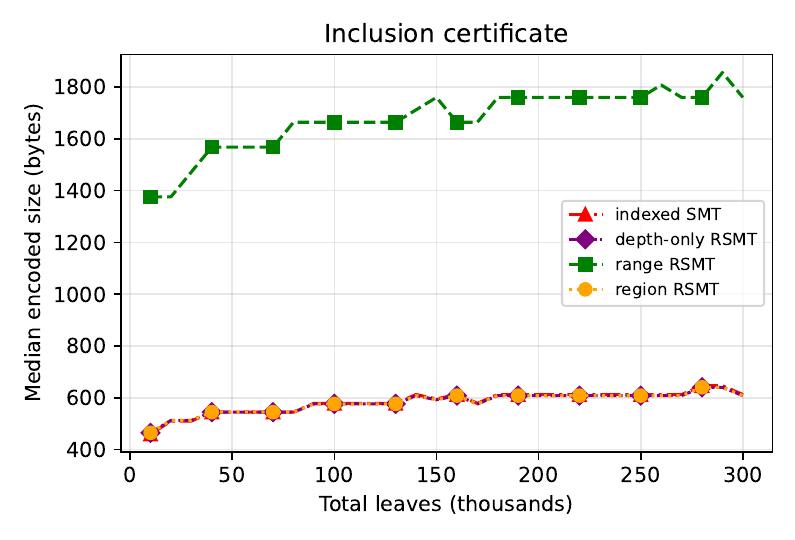}
    \caption{Median inclusion-certificate size, sampled over 32 inserted keys per round in the workload of Figure~\ref{fig:tree-design-consistency}. The region commitment preserves the compact depth-only format; authenticated ranges enlarge every sibling operand.}
    \label{fig:tree-design-inclusion}
\end{figure}

\subsection{Resulting Choice}

Figures~\ref{fig:tree-design-consistency} and~\ref{fig:tree-design-inclusion} show the relevant compromise. The indexed SMT remains a sound and useful baseline, while the depth-only radix construction marks the compactness attainable if placement is ignored. Authenticating subtree ranges restores soundness but makes the inclusion certificate substantially larger. Region commitments restore the same security with local, regular checks and leave inclusion certificates at the depth-only size. We therefore select the region-committing RSMT: it makes the untrusted update check arithmetization-friendly, preserves the shortest of the secure inclusion formats considered, and touches only the frontier affected by the current batch. The next section specifies this construction and its verifier.

\section{Consistency Proof}
\label{sec:consistency-proof}

A \emph{consistency proof} is a cryptographic construction that validates one round of operation of the append-only accumulator. Round $i$ inserts the batch $B_i = ((k_1, v_1), \ldots, (k_j, v_j))$ of key--value pairs into the tree; the root digest before the round is $r_{i-1}$, and after the round it is $r_i$. The consistency proof $\pi_i$ is a transcript of the part of the tree touched by the insertions. The Consensus Layer verifies $(\pi_i, r_{i-1}, r_i)$ before certifying $r_i$ (Algorithm~\ref{alg:consensuslayer}).

The property that verification enforces is \emph{append-only consistency} (Definition~\ref{def:minimal}; formally Definition~\ref{def:aoc}): every binding recorded under $r_{i-1}$ is recorded unchanged under $r_i$, and every new leaf sits at the position determined by its own key. Section~\ref{sec:stack-verifier} specifies the tree, the proof encoding, and the verifier; Section~\ref{sec:consistency-formal} develops the formal model of trees and certificates; Section~\ref{sec:single-key} first settles the single-insertion case, in which a non-inclusion certificate alone drives the update; Section~\ref{sec:consistency-theorem} then proves soundness and completeness of the batched form, assuming only collision resistance of the hash function.

\subsection{The Region-Committing Tree and Its Verifier}
\label{sec:stack-verifier}

The accumulator is implemented as the \emph{RSMT} (radix sparse Merkle tree): a binary Patricia tree over fixed-length keys, path-compressed, with a region-committing node hash. The consistency proof is a transcript of the touched part of the tree, executed by a stack machine.

Let $H \colon \{0,1\}^{*} \to \{0,1\}^{\lambda}$ be the hash function. A leaf hashes the full key together with the value, and a junction hashes its two children $x$ and $y$ together with its bifurcation depth $d$ and its \emph{region} $p$:
\begin{align*}
  \mathsf{LeafHash}(k,v) &:= H(\texttt{0x00} \parallel k \parallel v),\\
  \mathsf{NodeHash}(d,p,x,y) &:= H(\texttt{0x01} \parallel \langle d\rangle \parallel \langle p\rangle \parallel x \parallel y).
\end{align*}
The region $p \in \{0,1\}^d$ is the key prefix that addresses the node: every key below the junction extends $p$, with $p\|0$ leading left and $p\|1$ leading right; $\langle\cdot\rangle$ are fixed-length injective encodings. A leaf's region is its full key, and its depth is $\kappa = 256$. (We write $\varrho[j]$ for bit $j$ of a region and $\varrho[0..d)$ for its first $d$ bits.) The domain-separation prefixes make leaf and junction hashes disjoint; the depth commitment prevents re-attaching a subtree at a different level; the region commitment pins the node to its key-space position; and fixed child positions prevent swapping. Crucially, depth and region are \emph{absolute} properties of a node---splitting an edge above it changes neither---so inserting new keys never re-hashes any pre-existing node: an insertion creates only new leaf and junction hashes. This immutability of pre-state hashes is what the proof encoding exploits. (The hash function is an instantiation detail: the portable reference implementation uses SHA-256, the AIR a Poseidon2 sponge.) User-facing certificates transmit only junction depths and sibling digests, plus the leaf at which the path terminates; the verifier derives the junction regions from that leaf's key and binds the path to the target key by one bit comparison per junction. Inclusion and non-inclusion certificates are the same object, encoded as \texttt{InclusionCertificate} and differentiated by whether the terminal key equals the target key (Definition~\ref{def:inclusion-cert}). Read a second time, a non-inclusion certificate also determines the root that results from inserting $k$, which makes it a single-insertion consistency proof in its own right (Section~\ref{sec:single-key}).

The consistency proof $\pi$ for round $i$ is the post-order serialization of the \emph{touched} part of the post-state tree, over a five-opcode alphabet:
\begin{description}[nosep]
  \item[$S(c)$:] an untouched pre-state subtree, as an opaque digest;
  \item[$O(d', p', c_l, c_r)$:] an untouched pre-state junction, opened one level; the verifier hashes the opening, so the annotations are collision-bound to the digest;
  \item[$O_L(k', v')$:] an untouched pre-state leaf, opened;
  \item[$L$:] a leaf newly inserted in this round; its key and value are not part of the proof, but are consumed from the batch $B_i$;
  \item[$N(d)$:] a junction at bifurcation depth $d$, over the two preceding stack entries. Junction regions are not present in the proof; the verifier derives them.
\end{description}

The verifier (Algorithm~\ref{alg:stackverify}) executes the stream against a stack of triples: the subtree's pre-state digest, its post-state digest, and an \emph{advice tuple} $(\delta, \varrho)$---the depth and region of the subtree's top node, $(\kappa, k)$ for leaves, absent ($\bot$) for opaque $S$ entries. The batch is sorted by the verifier itself into tree-traversal order and must be strictly increasing, so the prover has no freedom in associating $L$ opcodes with batch elements. Processing $N(d)$ combines three rule families:
\begin{enumerate}[nosep]
  \item \emph{Edge coherence.} Every child that carries advice must satisfy $\delta > d$ and $\varrho[d] = \beta$, where $\beta$ is the child's side. The first $d$ bits of $\varrho$ yield the junction's region $p$; all advised children must agree on it, and at least one child must be advised, so $p$ is always defined.
  \item \emph{Confinement of opaque subtrees.} If the junction is absent from the pre-state---its old side arises by pass-through or $\varnothing$---then \emph{both} children must carry advice. An opaque $S$ may therefore appear only under pre-existing junctions, whose edges were checked in the round that created them and are frozen by the hashes; wherever a preserved subtree meets a new junction, the prover must present its opened form.
  \item \emph{Digest algebra.} The pre-state digest of a junction follows a four-way rule: if both children existed in the pre-state, the junction existed too and its old digest is recomputed; if exactly one child existed, the junction is new and the old digest of the existing child passes through unchanged; if neither existed, the old digest is the empty marker $\varnothing$. The post-state digest is always recomputed. The pass-through cases are what keep consistency proofs short: no hashing is performed on the parts of the tree that did not change.
\end{enumerate}

Collect the four-way pre-state rule in the helper
\begin{equation}
\mathsf{HashPreState}(d,p,x,y) =
  \begin{cases}
    \varnothing, & (x,y)=(\varnothing,\varnothing),\\
    y,           & x=\varnothing,\ y\neq\varnothing,\\
    x,           & x\neq\varnothing,\ y=\varnothing,\\
    \mathsf{NodeHash}(d,p,x,y), & x,y\neq\varnothing.
  \end{cases}
\label{eq:hash-pre-state}
\end{equation}
Thus a newly introduced junction disappears when reconstructing the pre-state, whereas a junction already present in that state is re-hashed.

\begin{algorithm}[t!]
  \caption{Stack-machine verification of the RSMT consistency proof}\label{alg:stackverify}
  \footnotesize
  \begin{algorithmic}[0]
    \Function{VerifyConsistency}{$\pi, r_{i-1}, r_i, B$}
      \If{$B = [\,]$}
        \State \Return $r_{i-1} = r_i \;\land\; \pi = [\,]$
      \EndIf
      \State $B \gets \Call{SortTraversalOrder}{B}$
      \State \textbf{assert} keys of $B$ in $\{0,1\}^{\kappa}$, strictly increasing
      \State $\mathit{st} \gets [\,]$; \quad $b \gets 0$ \Comment{stack; batch index}
      \For{opcode $o$ \textbf{in} $\pi$}
        \If{$o = S(c)$} \Comment{opaque subtree}
          \State \textbf{assert} $c \in \{0,1\}^{\lambda}$
          \State \Call{Push}{$\mathit{st}, (c, c, \bot)$}
        \ElsIf{$o = O(d', p', c_l, c_r)$} \Comment{opening}
          \State \textbf{assert} $0 \le d' < \kappa \land p' \in \{0,1\}^{d'} \land c_l, c_r \in \{0,1\}^{\lambda}$
          \State $c \gets \Call{NodeHash}{d', p', c_l, c_r}$
          \State \Call{Push}{$\mathit{st}, (c, c, (d', p'))$}
        \ElsIf{$o = O_L(k', v')$} \Comment{opened leaf}
          \State \textbf{assert} $k' \in \{0,1\}^{\kappa}$
          \State $c \gets \Call{LeafHash}{k', v'}$
          \State \Call{Push}{$\mathit{st}, (c, c, (\kappa, k'))$}
        \ElsIf{$o = L$} \Comment{new leaf}
          \State $(k, v) \gets B[b]$; \quad $b \gets b + 1$
          \State \Call{Push}{$\mathit{st}, (\varnothing,\; \Call{LeafHash}{k, v},\; (\kappa, k))$}
        \ElsIf{$o = N(d)$} \Comment{junction}
          \State \textbf{assert} $0 \le d < \kappa$
          \State $(c^{\mathsf o}_r, c^{\mathsf n}_r, a_r) \gets \Call{Pop}{\mathit{st}}$
          \State $(c^{\mathsf o}_l, c^{\mathsf n}_l, a_l) \gets \Call{Pop}{\mathit{st}}$
          \State $p \gets \bot$
          \For{$x \in \{l, r\}$ with side bit $\beta \in \{0, 1\}$}
            \If{$a_x = (\delta_x, \varrho_x) \neq \bot$} \Comment{coherence}
              \State \textbf{assert} $\delta_x > d \;\land\; \varrho_x[d] = \beta$
              \State \textbf{assert} $p = \bot \;\lor\; p = \varrho_x[0..d)$
              \State $p \gets \varrho_x[0..d)$
            \EndIf
          \EndFor
          \State \textbf{assert} $p \neq \bot$
          \If{$c^{\mathsf o}_l = \varnothing \lor c^{\mathsf o}_r = \varnothing$} \Comment{new junction}
            \State \textbf{assert} $a_l \neq \bot \;\land\; a_r \neq \bot$
          \EndIf
          \State $c^{\mathsf o} \gets \Call{HashPreState}{d,p,c^{\mathsf o}_l,c^{\mathsf o}_r}$
          \State $c^{\mathsf n} \gets \Call{NodeHash}{d,p,c^{\mathsf n}_l,c^{\mathsf n}_r}$
          \State \Call{Push}{$\mathit{st}, (c^{\mathsf o}, c^{\mathsf n}, (d, p))$}
        \Else
          \State \Return $0$ \Comment{unknown opcode}
        \EndIf
      \EndFor
      \State \textbf{assert} $b = |B| \;\land\; |\mathit{st}| = 1$
      \State \Return $\mathit{st}[0].(c^{\mathsf o}, c^{\mathsf n}) = (r_{i-1}, r_i)$
    \EndFunction
  \end{algorithmic}
\end{algorithm}

Verification accepts iff the opcode stream and the batch are both fully consumed, the stack holds exactly one triple, and its digest pair equals $(r_{i-1}, r_i)$; any failed assertion, malformed opcode, stack underflow, or batch overrun rejects. The verifier is a short loop with no recursion and no control flow beyond opcode dispatch---one uniform rule per opcode; its memory use is bounded by the tree depth, and the post-order format makes verification a natural streaming computation. This regularity is deliberate: it is what the AIR arithmetization of Section~\ref{sec:custom-air-circuit} exploits, one trace row per opcode with a fixed constraint family. In-circuit, the equations that derive a junction region from its children are also the edge-coherence constraints, and the same canonical region limbs feed the junction hash. The arithmetization implements this complete rule set, including operand-domain checks, canonical encodings, confinement at new junctions, and coherent key placement.

For the honest prover, the transcript differs from a bare depth-only encoding only in the openings: one opened junction or leaf per split edge---the node the insertion descends past last. Measured with the reference implementation on a batch of $1\,000$ insertions into a tree of $10^4$ keys, the transcript is within $15\%$ of the size of the depth-only encoding; the derived regions cost nothing on the wire.

\subsection{Insertion Example}
\label{sec:worked-insertion}

Consider the two-leaf pre-state in Figure~\ref{fig:worked-insertion}. The root junction separates $a$ from $b$ at depth $d_2$. A new leaf $c$ is added with key $k_c$; it follows the route towards $b$, shares the longer prefix $p_1$ with $k_b$, and diverges from it at depth $d_1$, where $d_2<d_1<\kappa$. Thus $k_a[d_2]=0$, $k_b[d_2]=k_c[d_2]=1$, while $k_b[d_1]=0$ and $k_c[d_1]=1$. Inserting $c=(k_c,v_c)$ splits the edge above $b$: it creates a new leaf and a new junction $N(d_1,p_1)$, but changes no pre-existing leaf hash.

\begin{figure}[htb]
  \centering
  \begin{tikzpicture}[
    diagram,
    junc/.style={diagram box, minimum width=2.25cm, minimum height=0.75cm},
    leaf/.style={diagram box, minimum width=2.15cm, minimum height=0.75cm},
    preserved/.style={fill=blue!7},
    opened/.style={fill=yellow!18},
    fresh/.style={fill=green!12},
    recomputed/.style={fill=orange!12},
    treeedge/.style={diagram arrow}
  ]
    % ---------- Pre-state (left) ----------
    \node[font=\bfseries] at (0,5.6) {Pre-state: $r_{i-1}$};
    \node[junc, preserved] (oldroot) at (0,4.5) {$N(d_2,p_2)$};
    \node[leaf, preserved] (olda) at (-1.55,3.0) {Leaf $a$\\$h_a$};
    \node[leaf, preserved] (oldb) at (1.55,3.0) {Leaf $b$\\$h_b$};
    \draw[treeedge] (oldroot) -- node[above left] {$0$} (olda);
    \draw[treeedge] (oldroot) -- node[above right] {$1$} (oldb);

    % ---------- Transition arrow ----------
    \node at (4.4,4.0) {$\Rightarrow$};
    \node[align=center] at (4.4,3.4) {insert\\$c=(k_c,v_c)$};

    % ---------- Post-state (right) ----------
    \node[font=\bfseries] at (8.8,5.6) {Post-state: $r_i$};
    \node[junc, recomputed] (newroot) at (8.8,4.5) {Recomputed node\\$N(d_2,p_2)$};
    \node[leaf, preserved] (newa) at (7.05,3.0) {Leaf $a$\\$S(h_a)$};
    \node[junc, fresh] (newj) at (10.55,3.0) {New node\\$N(d_1,p_1)$};
    \node[leaf, preserved] (newb) at (9.45,1.35) {Leaf $b$\\$O_L(k_b,v_b)$};
    \node[leaf, fresh] (newc) at (11.65,1.35) {New leaf $c$\\$L$};
    \draw[treeedge] (newroot) -- node[above left] {$0$} (newa);
    \draw[treeedge] (newroot) -- node[above right] {$1$} (newj);
    \draw[treeedge] (newj) -- node[above left] {$0$} (newb);
    \draw[treeedge] (newj) -- node[above right] {$1$} (newc);
  \end{tikzpicture}
  \caption{Inserting $c$ splits the edge above $b$. Leaf $a$ is opaque because it remains below the pre-existing junction $N(d_2,p_2)$. Leaf $b$ must be opened because it becomes a child of the new junction $N(d_1,p_1)$; this supplies the advice needed to check the new edge.}
  \label{fig:worked-insertion}
\end{figure}
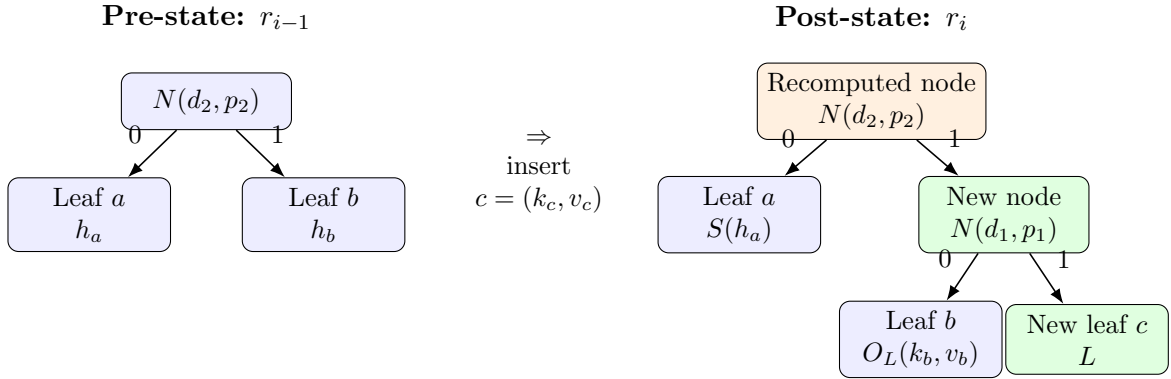

Let $h_x=\mathsf{LeafHash}(k_x,v_x)$ for $x\in\{a,b,c\}$ and
\[
  h_{N_1}=\mathsf{NodeHash}(d_1,p_1,h_b,h_c).
\]
The old and new roots are, respectively,
\begin{align*}
  r_{i-1}&=\mathsf{NodeHash}(d_2,p_2,h_a,h_b),\\
  r_i&=\mathsf{NodeHash}(d_2,p_2,h_a,h_{N_1}).
\end{align*}
Post-order traversal gives the five-opcode transcript
\[
  \pi=\bigl[S(h_a),\ O_L(k_b,v_b),\ L,\ N(d_1),\ N(d_2)\bigr].
\]
The leaf $b$ is an opening rather than an opaque $S$: it meets a junction absent from the pre-state, so confinement requires both children of that junction to carry advice. Leaf $a$ remains opaque because the outer junction existed previously.

Table~\ref{tab:worked-stack} traces Algorithm~\ref{alg:stackverify} over the transcript: for every opcode it lists the checks performed, the pre-state and post-state digest computations, and the stack contents afterwards. The asymmetry of the two readings is visible in the junction rows: the post-state side hashes every junction, whereas on the pre-state side the new junction $N(d_1)$ contributes no hash---$\mathsf{HashPreState}$ passes $h_b$ through, so the pre-state root is assembled exactly as if leaf $b$ still hung on the outer edge. The advice fields make the placement checks concrete: the opened leaf $b$ supplies the advice that the new junction $N(d_1)$ requires, the opaque $S(h_a)$ supplies none, and $N(d_2)$ accepts it only because both of its old child digests are present.

\begin{table}[H]
\centering
\caption{Stack-machine verification of the insertion in Figure~\ref{fig:worked-insertion}. Each stack entry is a triple $(c^{\mathsf o}, c^{\mathsf n}, a)$ of pre-state digest, post-state digest, and advice $(\delta,\varrho)$; entries are listed bottom to top, and a junction opcode pops the two topmost entries as its left and right child.}
\label{tab:worked-stack}
\footnotesize
\renewcommand{\arraystretch}{1.3}
\setlength{\tabcolsep}{4pt}
\begin{tabular}{@{}
  >{\raggedright\arraybackslash}p{1.6cm}
  >{\raggedright\arraybackslash}p{3.5cm}
  >{\raggedright\arraybackslash}p{3.5cm}
  >{\raggedright\arraybackslash}p{3.5cm}
  >{\raggedright\arraybackslash}p{3.1cm}@{}}
\toprule
\textbf{Opcode} & \textbf{Checks} & \textbf{Pre-state side $c^{\mathsf o}$} & \textbf{Post-state side $c^{\mathsf n}$} & \textbf{Stack after} \\
\midrule
$S(h_a)$ &
  $h_a \in \{0,1\}^{\lambda}$ &
  $h_a$, opaque; not recomputed &
  $h_a$; the subtree is untouched &
  $(h_a,\,h_a,\,\bot)$ \\
\addlinespace
$O_L(k_b,v_b)$ &
  $k_b \in \{0,1\}^{\kappa}$ &
  $h_b = \mathsf{LeafHash}(k_b,v_b)$, recomputed from the opened preimage &
  $h_b$; the leaf persists &
  \makecell[l]{$(h_a,\,h_a,\,\bot)$\\ $(h_b,\,h_b,\,(\kappa,k_b))$} \\
\addlinespace
$L$ &
  consumes $(k_c,v_c)$, the next element of the sorted batch &
  $\varnothing$; no leaf at $k_c$ in the pre-state &
  $h_c = \mathsf{LeafHash}(k_c,v_c)$ &
  \makecell[l]{$(h_a,\,h_a,\,\bot)$\\ $(h_b,\,h_b,\,(\kappa,k_b))$\\ $(\varnothing,\,h_c,\,(\kappa,k_c))$} \\
\addlinespace
$N(d_1)$ \newline pops $b$, $c$ &
  coherence: $\kappa > d_1$, $k_b[d_1]=0$, $k_c[d_1]=1$; agreement: $p_1 = k_b[0..d_1) = k_c[0..d_1)$; confinement: $c^{\mathsf o}$ of $c$ is $\varnothing$, so both children must be advised---hence $b$ was sent opened &
  $\mathsf{HashPreState}$\newline${}(d_1,p_1,h_b,\varnothing) = h_b$; the junction is absent from the pre-state and the surviving child passes through &
  $\mathsf{NodeHash}(d_1,p_1,h_b,h_c)$\newline${} = h_{N_1}$; the new junction is hashed &
  \makecell[l]{$(h_a,\,h_a,\,\bot)$\\ $(h_b,\,h_{N_1},\,(d_1,p_1))$} \\
\addlinespace
$N(d_2)$ \newline pops $a$, $N_1$ &
  coherence: $d_1 > d_2$, $p_1[d_2]=1$; region: $p_2 = p_1[0..d_2)$; the unadvised left child is permitted, as both old digests are present and the junction pre-exists &
  both children present, so the pre-state junction is re-hashed: \newline $\mathsf{NodeHash}(d_2,p_2,h_a,h_b)$\newline${} = r_{i-1}$ &
  $\mathsf{NodeHash}(d_2,p_2,h_a,h_{N_1})$\newline${} = r_i$ &
  $(r_{i-1},\,r_i,\,(d_2,p_2))$ \\
\bottomrule
\end{tabular}
\end{table}

The proof accepts because the transcript and batch are exhausted and the final stack consists of the single entry $(r_{i-1},r_i,(d_2,p_2))$. Its first two fields are compared directly with the certified previous and current roots.

\begin{remark}[Why the region commitment is necessary]
\label{rem:placement}
Everything in Algorithm~\ref{alg:stackverify} beyond the digest algebra exists to enforce coherent placement. Consider the minimal alternative: the junction hash commits the depth alone, $h_N = H(\texttt{0x01} \| \langle d\rangle \| h_l \| h_r)$, the alphabet shrinks to $S$, $L$, $N$, and only the digest algebra is checked. Let the pre-state with certified root $r_{i-1}$ contain the recorded binding $(k, v)$, and pick any depth $d^{*}$ with $k[d^{*}] = 1$. The three-opcode stream $(S(r_{i-1}), L, N(d^{*}))$ then verifies for the batch $\{(k, v')\}$, $v' \ne v$: the old side of $N(d^{*})$ passes $r_{i-1}$ through, and the new side hangs the entire pre-state tree on the $0$-side of the new junction, while the key-directed descent for $k$ leads to the $1$-side. Under the certified $r_i$, a one-step inclusion proof for $(k, v')$ verifies, and the preserved binding $(k, v)$ is no longer provable: cross-round equivocation on $k$, invisible to every structural check, including the full-history audit of Section~\ref{sec:aggregation-audit}. Algorithm~\ref{alg:stackverify} rejects the stream: the preserved child of a new junction must be presented opened, and edge coherence requires its region to extend $p\|0$ while the new leaf's key extends $p\|1$---impossible, since the preserved subtree's region is a prefix of $k$. Theorem~\ref{thm:round} shows that every attack of this kind is excluded.
\end{remark}

\subsection{Formal Model}
\label{sec:consistency-formal}

We fix the key length $\kappa = 256$, key space $\mathcal{K} = \{0,1\}^{\kappa}$, a value space $\mathcal{V} \subseteq \{0,1\}^{*}$, and a hash function $H \colon \{0,1\}^{*} \to \{0,1\}^{\lambda}$. For bit strings, $p \preceq q$ denotes that $p$ is a prefix of $q$, $q[j]$ is the $j$-th bit, $q[0..d)$ is the first $d$ bits, and $\mathsf{lcp}(x,y)$ is the longest common prefix of $x$ and $y$; $\langle\cdot\rangle$ are fixed-length injective encodings; $\varnothing$ is a distinguished constant outside $\{0,1\}^{\lambda}$. All statements below are unconditional reductions: each concludes either the stated property, or that two distinct strings with equal $H$-images (a \emph{collision}) are computable in time linear in the size of the objects at hand. We do not repeat this disjunction in every statement.

\subsubsection{Tree commitments}

\begin{definition}[Valid trees]
\label{def:tree}
Trees are generated by
\[ T ::= \mathsf{Leaf}(k, v) \mid \mathsf{Node}(d, p, T_l, T_r) \]
with $k \in \mathcal{K}$, $v \in \mathcal{V}$, $0 \le d < \kappa$, $p \in \{0,1\}^{d}$; $\varepsilon$ denotes the empty tree. Digests:
\begin{align*}
\mathsf{dig}(\mathsf{Leaf}(k,v)) &= \mathsf{LeafHash}(k,v),\\
\mathsf{dig}(\mathsf{Node}(d,p,l,r)) &= \mathsf{NodeHash}(d,p,\mathsf{dig}(l),\mathsf{dig}(r)),
\end{align*}
and $\mathsf{dig}(\varepsilon) = \varnothing$. Region and depth: $\mathsf{reg}(\mathsf{Leaf}(k,v)) = k$ and $\mathsf{reg}(\mathsf{Node}(d,p,\cdot,\cdot)) = p$; $\mathsf{dep}(\mathsf{Leaf}) = \kappa$ and $\mathsf{dep}(\mathsf{Node}(d,\ldots)) = d$. A tree is \emph{valid} if every junction $\mathsf{Node}(d,p,l,r)$ in it satisfies: $l$ and $r$ are nonempty, $\mathsf{dep}(l) > d$, $\mathsf{dep}(r) > d$, $p\|0 \preceq \mathsf{reg}(l)$, and $p\|1 \preceq \mathsf{reg}(r)$. The trees $\varepsilon$ and $\mathsf{Leaf}(k,v)$ are valid.
\end{definition}

The validity conditions are local: one condition per edge. The next three results separate their global consequences from the stronger persistence facts needed only for completeness. For a finite partial map $M$ and a bit string $q$, write
\[
  M_q = \{\,k \mapsto v \in M \mid q \preceq k\,\}
\]
for the restriction of $M$ to the key-space cone below $q$.

\begin{lemma}[Local-to-global shape]
\label{lem:shape}
Let $T$ be a valid tree. Then:
\begin{enumerate}[nosep,label=(\roman*)]
  \item every leaf key extends the region of each of its ancestors, and distinct leaves carry distinct keys; hence $T$ represents the partial map $\mathsf{map}(T) = \{k \mapsto v \mid \mathsf{Leaf}(k,v) \in T\}$;
  \item at every junction, the region $p$ is the longest common prefix of the keys below it;
  \item the left-to-right leaf order of $T$ is the strictly increasing key order;
  \item every subtree $U$ of $T$ with region $\varrho = \mathsf{reg}(U)$ satisfies $\mathsf{map}(U) = \mathsf{map}(T)_\varrho$ (the \emph{cone identity}).
\end{enumerate}
\end{lemma}
\begin{proof}
(i) Validity gives $p\|\beta \preceq \mathsf{reg}(\text{child})$ at every edge, so by induction every node's region, and every leaf key, extends the region of each ancestor. Two leaves with the same key $k$ would sit on opposite sides of their lowest common ancestor $\mathsf{Node}(d,p,\cdot,\cdot)$, forcing both $p\|0 \preceq k$ and $p\|1 \preceq k$, which is impossible.

(ii) Every key below $\mathsf{Node}(d,p,l,r)$ extends $p$. Both children are nonempty, so some key below extends $p\|0$ and some extends $p\|1$. The common prefix therefore ends after exactly the $d$ bits of $p$.

(iii) At every junction, keys on the left have bit $0$ and keys on the right have bit $1$ at position $d$, and both sides agree on the first $d$ bits. So every left key precedes every right key, and the claim follows by induction.

(iv) By (i), every key below $U$ extends $\varrho$, so $\mathsf{map}(U) \subseteq \mathsf{map}(T)_\varrho$. Conversely, suppose a leaf with key $k$, $\varrho \preceq k$, lies outside $U$. The lowest common ancestor of that leaf and $U$ is a junction at some depth $d < \mathsf{dep}(U) = |\varrho|$, and by (i) applied to its two sides, $\varrho$ and $k$ differ at bit $d$, contradicting $\varrho \preceq k$.
\end{proof}

\begin{proposition}[Unique representation]
\label{prop:unique}
Every finite partial map $M$ has exactly one valid RSMT, written $\mathsf{Tree}(M)$.
\end{proposition}
\begin{proof}
Induction on the key set. The empty map gives $\varepsilon$, and a singleton forces its leaf. Otherwise let $p$ be the longest common prefix of all keys. Its next bit partitions $M$ into the two nonempty maps $M_{p\|0}$ and $M_{p\|1}$. By induction these maps have unique valid trees $T_l$ and $T_r$. Their root regions extend $p\|0$ and $p\|1$, respectively, and their root depths exceed $|p|$, so $\mathsf{Node}(|p|,p,T_l,T_r)$ is valid.

Conversely, Lemma~\ref{lem:shape}(ii) forces any valid root representing $M$ to have region $p$; validity forces its two children to represent exactly $M_{p\|0}$ and $M_{p\|1}$; and the induction hypothesis forces those children. Values label the determined leaves and do not affect the shape.
\end{proof}

\begin{lemma}[Extension persistence]
\label{lem:persist}
Let $M$ and $B$ be finite partial maps with disjoint domains. Then:
\begin{enumerate}[nosep,label=(\roman*)]
  \item $\mathsf{Tree}(M)$ contains a junction with region $p$ iff $M_{p\|0} \neq \emptyset \neq M_{p\|1}$;
  \item every junction of $\mathsf{Tree}(M)$ persists, with the same depth and region, in $\mathsf{Tree}(M \uplus B)$;
  \item if a node with region $\varrho$ occurs in either $\mathsf{Tree}(M)$ or $\mathsf{Tree}(M \uplus B)$ and $B_\varrho = \emptyset$, then it occurs in both trees and the two rooted subtrees are identical.
\end{enumerate}
\end{lemma}
\begin{proof}
(i) For the forward direction, the cone identity (Lemma~\ref{lem:shape}(iv)) and the two nonempty children supply keys extending $p\|0$ and $p\|1$. Conversely, suppose both restricted maps are nonempty and descend from the root of $\mathsf{Tree}(M)$. At a current junction with region $q$, the two selected key sets lie below it, so Lemma~\ref{lem:shape}(ii) gives $q \preceq p$. If $q \neq p$, all keys extending $p$ select the same child at depth $|q|$; recurse into that child. Depths increase, the two selected sets prevent termination at a leaf, and the descent therefore reaches the junction with region $p$.

(ii) The condition in (i) is monotone under extension from $M$ to $M \uplus B$, and a junction's depth is the length of its region.

(iii) Suppose first that the node occurs in $\mathsf{Tree}(M)$. A junction persists by (ii), while a leaf persists because its binding remains in the extended map. Conversely, suppose the node occurs in $\mathsf{Tree}(M \uplus B)$. If it is a junction, the two sides of $(M \uplus B)_\varrho$ are nonempty; because $B_\varrho=\emptyset$, clause (i) places the same junction in $\mathsf{Tree}(M)$. If it is a leaf, its binding belongs to $M$ and the same leaf occurs in $\mathsf{Tree}(M)$. In either direction, by the cone identity the two rooted subtrees represent the same map
\[
  (M \uplus B)_\varrho=M_\varrho.
\]
Proposition~\ref{prop:unique} makes them identical.
\end{proof}

For a valid tree $T$ and a batch $B$ whose keys are pairwise distinct and absent from $\mathsf{map}(T)$, we identify $B$ with its induced partial map and write
\[
  T \oplus B = \mathsf{Tree}(\mathsf{map}(T) \uplus B).
\]

\begin{definition}[Opening trees]
\label{def:opening}
Opening trees are generated by
\[ F ::= \mathsf{Hole}(c) \mid \mathsf{OLeaf}(k, v) \mid \mathsf{ONode}(d, p, F_l, F_r) \]
with $c \in \{0,1\}^{\lambda}$ and the other operands as in Definition~\ref{def:tree}. Evaluation mirrors $\mathsf{dig}$: writing $e_x = \mathsf{eval}(F_x)$,
\begin{align*}
\mathsf{eval}(\mathsf{Hole}(c)) &= c,\\
\mathsf{eval}(\mathsf{OLeaf}(k,v)) &= \mathsf{LeafHash}(k,v),\\
\mathsf{eval}(\mathsf{ONode}(d,p,F_l,F_r)) &= \mathsf{NodeHash}(d,p,e_l,e_r).
\end{align*}
\end{definition}

An opening tree is a partially opened commitment: $\mathsf{OLeaf}$ and $\mathsf{ONode}$ present hash preimages, while a $\mathsf{Hole}$ stands for an unopened subtree. The following lemma is the workhorse of every argument below: an opening tree that evaluates to the digest of a known tree is an exact partial copy of that tree.

\begin{lemma}[Matching]
\label{lem:match}
Let $T \neq \varepsilon$ be a tree and $F$ an opening tree with $\mathsf{eval}(F) = \mathsf{dig}(T)$. Then:
\begin{enumerate}[nosep,label=(\alph*)]
  \item every $\mathsf{ONode}(d,p,\cdot,\cdot)$ of $F$ coincides with a junction $\mathsf{Node}(d,p,\cdot,\cdot)$ of $T$ at the same position, with equal child digests;
  \item every $\mathsf{OLeaf}(k,v)$ of $F$ coincides with $\mathsf{Leaf}(k,v)$ of $T$ at the same position;
  \item every $\mathsf{Hole}(c)$ is \emph{assigned} the subtree of $T$ at its position, and that subtree has digest $c$;
  \item replacing every hole by its assigned subtree turns $F$ into $T$.
\end{enumerate}
\end{lemma}
\begin{proof}
Induction on $F$, with the invariant $\mathsf{eval}(F) = \mathsf{dig}(T)$ for the current pair. A hole is assigned the current subtree of $T$. Otherwise, compare the two $H$-preimages. If they differ, they are a collision. If they are equal, the domain-separation tag and the injective encodings force the same constructor with equal components: an $\mathsf{OLeaf}$ meets $\mathsf{Leaf}(k,v)$ with the same $k$ and $v$; an $\mathsf{ONode}$ meets a junction with the same $(d,p)$ and equal child digests, and the induction continues in both children. Clause (d) restates that the recursion covers $F$ and $T$ simultaneously and completely. Each step is one comparison, so the reduction is linear.
\end{proof}

\subsubsection{Inclusion and non-inclusion certificates}

Certificates concern one committed map and do not depend on how that map was reached. We therefore define and prove them before introducing certified update histories.

\begin{definition}[Inclusion certificate encoding]
\label{def:inclusion-cert}
Inclusion and non-inclusion are certified by one object. Its encoded data-structure name is \texttt{InclusionCertificate}; the name denotes the format, not the relation. A non-empty certificate is a possibly empty sequence of junction openings followed by a terminal leaf,
\[
  C^\mathsf{inc}=\big((d_1,s_1),\ldots,(d_m,s_m);\ (k',v')\big), \qquad m\geq 0,
\]
with $k'\in\mathcal{K}$, $v'\in\mathcal{V}$ and every $s_j\in\{0,1\}^{\lambda}$. When $m>0$, the depths satisfy
\[
  \kappa>d_1>\cdots>d_m\geq 0,
\]
and the order is from the leaf toward the root. There is in addition one distinguished empty certificate $C^\mathsf{inc}_\emptyset$.

To verify $C^\mathsf{inc}$ for a target key $k$ against $r$, set $c_0=\mathsf{LeafHash}(k',v')$ and, for $j=1,\ldots,m$, require $k[d_j]=k'[d_j]$ and compute
\[
c_j = \begin{cases}
\mathsf{NodeHash}(d_j,p_j,c_{j-1},s_j) & \text{if } k'[d_j] = 0,\\
\mathsf{NodeHash}(d_j,p_j,s_j,c_{j-1}) & \text{if } k'[d_j] = 1,
\end{cases}
\]
where $p_j=k'[0..d_j)$. Accept iff every requirement holds and $c_m=r$; when $m=0$, this means accepting iff $c_0=r$. The empty certificate $C^\mathsf{inc}_\emptyset$ accepts iff $r=\varnothing$. Any unmet requirement rejects.

An accepted certificate is an \emph{inclusion certificate} for $(k,v)$ if $(k',v')=(k,v)$, and a \emph{non-inclusion certificate} for $k$ if $k'\neq k$; the empty certificate is a non-inclusion certificate for every $k$. The openings are read identically in both cases, and only the terminal comparison distinguishes them.

The junction regions are reconstructed from the \emph{terminal} key $k'$, instead of the target key $k$: past the recorded key space, the regions on the path are no longer prefixes of $k$, and deriving them from $k$ would fail. The path is tied to $k$ by the per-level requirement $k[d_j]=k'[d_j]$, which shows that the key-directed descent for $k$---at a junction $\mathsf{Node}(d,p,l,r)$, enter the child on side $k[d]$---traverses these openings. Since that descent is deterministic and lands on $\mathsf{Leaf}(k,\cdot)$ whenever $k$ is recorded, reaching any other leaf is a proof of absence.

For later use, $C^\mathsf{inc}$ determines an opening tree $F_{C^\mathsf{inc}}$. Start with $F_0=\mathsf{OLeaf}(k',v')$ and wrap it once per pair $(d_j,s_j)$:
\[
F_j=\begin{cases}
\mathsf{ONode}(d_j,p_j,F_{j-1},\mathsf{Hole}(s_j)) & k'[d_j]=0,\\
\mathsf{ONode}(d_j,p_j,\mathsf{Hole}(s_j),F_{j-1}) & k'[d_j]=1.
\end{cases}
\]
Thus $\mathsf{eval}(F_{C^\mathsf{inc}})=c_m$, so the digest part of verification is exactly the check $\mathsf{eval}(F_{C^\mathsf{inc}})=r$.
\end{definition}

\begin{theorem}[Certificate correctness]
\label{thm:inclusion-cert}
Let $M$ be a finite partial map, $T=\mathsf{Tree}(M)$, and $r=\mathsf{dig}(T)$.
\begin{enumerate}[nosep,label=(\roman*)]
  \item Let a certificate verify for $k$ against $r$. If it is empty, then $M$ is empty. Otherwise, if its terminal is $(k',v')$, then $M(k')=v'$; consequently it certifies $M(k)=v'$ when $k'=k$, and $k\notin\mathrm{dom}(M)$ when $k'\neq k$.
  \item If $M(k)=v$, the leaf-to-root path of $\mathsf{Leaf}(k,v)$ in $T$ gives a verifying inclusion certificate for $(k,v)$. If $k\notin\mathrm{dom}(M)$, key-directed descent in $T$ gives a verifying non-inclusion certificate for $k$; for $T=\varepsilon$ this is the empty certificate.
\end{enumerate}
Both certificates contain at most $\kappa$ junction openings and are constructed in time linear in the path length.
\end{theorem}
\begin{proof}
(i) The empty certificate verifies only when $r=\varnothing$, which means $T=\varepsilon$ and $\mathrm{dom}(M)=\emptyset$, so the absence claim holds. Let $C^\mathsf{inc}$ be nonempty and verify. Its opening tree satisfies $\mathsf{eval}(F_{C^\mathsf{inc}})=r$; since an opening tree evaluates to an $H$-image, $r\neq\varnothing$ and $T\neq\varepsilon$. Lemma~\ref{lem:match} places the terminal $\mathsf{OLeaf}(k',v')$ at the corresponding leaf of $T$, giving $M(k')=v'$, and places every $\mathsf{ONode}(d_j,p_j,\cdot,\cdot)$ at the junction of $T$ at the same position. The openings are therefore exactly the ancestors of $\mathsf{Leaf}(k',v')$, at depths $d_j$, and the path descends on side $k'[d_j]$ at each.

Suppose $k\in\mathrm{dom}(M)$. Every ancestor $\mathsf{Node}(d,p,l,r)$ of $\mathsf{Leaf}(k,M(k))$ has $p\preceq k$ by Lemma~\ref{lem:shape}(i), and the child holding that leaf has region a prefix of $k$; validity then forces that child to be the one on side $k[d]$. Hence key-directed descent for $k$ ends at $\mathsf{Leaf}(k,M(k))$. The requirement $k[d_j]=k'[d_j]$ makes that same descent traverse the openings, so it ends at $\mathsf{Leaf}(k',v')$. Descent is deterministic, so $k'=k$. Contrapositively, $k'\neq k$ gives $k\notin\mathrm{dom}(M)$.

(ii) Suppose $M(k)=v$. Starting at $\mathsf{Leaf}(k,v)$, list its ancestors toward the root, recording for each its depth and the digest of the sibling subtree, and take $(k,v)$ as the terminal. Depths strictly decrease in this leaf-to-root order, Lemma~\ref{lem:shape}(ii) gives the region $k[0..d)$, and the side is $k[d]$; with $k'=k$ the requirement $k[d_j]=k'[d_j]$ is trivial. The verifier therefore recomputes the digest of each ancestor and ends at $r$. A singleton tree gives $m=0$.

Suppose $k$ is absent. If $T=\varepsilon$, use the empty certificate. Otherwise follow the key-directed descent for $k$ to the leaf $\mathsf{Leaf}(k',v')$ at which it arrives, and record the same data. The descent is finite because junction depths strictly increase, and every node it passes is an ancestor of that leaf, so the recorded openings satisfy every digest check; the sides it takes are $k[d_j]$, and they are the path's own sides $k'[d_j]$, so the per-level requirement holds. Finally $k'\neq k$, since otherwise $k\in\mathrm{dom}(M)$.

A valid root-to-leaf path has strictly increasing depths in $\{0,\ldots,\kappa\}$, which gives the size bound; both constructions visit each path node once.
\end{proof}

\subsection{Single-Insertion Updates}
\label{sec:single-key}

Reduced to its simplest form, a round of the accumulator records one binding: the operator submits $(k,v)$ together with a non-inclusion certificate for $k$, the verifier checks the certificate against the current root, and the next root is read off the same certificate. This subsection develops that protocol into a complete security argument. It needs nothing beyond the certificates just defined: the post-state is \emph{computed} from the pre-state rather than supplied by the prover, so the argument never has to check a structural property of prover-supplied data, and the region commitment of the node hash carries no weight in it (Remark~\ref{rem:derived-vs-checked}). The batch transcript of Section~\ref{sec:stack-verifier} is then a compression of this protocol, paid for with the placement checks whose sufficiency Section~\ref{sec:consistency-theorem} proves (Remark~\ref{rem:single-key-border}).

\begin{definition}[Insertion]
\label{def:insert}
For $k \in \mathcal{K}$ and $v \in \mathcal{V}$, define $\mathsf{ins}(T;k,v)$ on valid trees:
\begin{enumerate}[nosep,label=(\arabic*)]
  \item $\mathsf{ins}(\varepsilon;k,v) = \mathsf{Leaf}(k,v)$;
  \item $\mathsf{ins}(\mathsf{Leaf}(k,v');k,v) = \mathsf{Leaf}(k,v)$;
  \item if $T \neq \varepsilon$ and $\varrho = \mathsf{reg}(T) \not\preceq k$: with $\hat p = \mathsf{lcp}(\varrho,k)$ and $\hat d = |\hat p|$,
        \[
          \mathsf{ins}(T;k,v) =
          \begin{cases}
            \mathsf{Node}\bigl(\hat d,\hat p,\mathsf{Leaf}(k,v),T\bigr) & k[\hat d]=0,\\
            \mathsf{Node}\bigl(\hat d,\hat p,T,\mathsf{Leaf}(k,v)\bigr) & k[\hat d]=1;
          \end{cases}
        \]
  \item if $T = \mathsf{Node}(d,p,l,r)$ and $p \preceq k$: $\mathsf{ins}(T;k,v) = \mathsf{Node}(d,p,\mathsf{ins}(l;k,v),r)$ if $k[d]=0$, and $\mathsf{Node}(d,p,l,\mathsf{ins}(r;k,v))$ if $k[d]=1$.
\end{enumerate}
\end{definition}

The case split is exhaustive and exclusive: for a leaf, $\mathsf{reg} \preceq k$ forces the key $k$ itself (clause~2), and otherwise clause~3 applies; for a junction, $|p| = d < \kappa$ makes exactly one of clauses~3 and~4 fire. In clause~3, $\varrho \not\preceq k$ makes $\hat p$ a proper prefix of $\varrho$, so $\hat d < \mathsf{dep}(T) \le \kappa$ and the bit $k[\hat d]$ exists; by maximality of $\hat p$, $\varrho[\hat d] = 1 - k[\hat d]$, so both new edges satisfy the validity conditions of Definition~\ref{def:tree}. Clause~2 is the only one that overwrites a value; Corollary~\ref{cor:run-append} shows that it is never reached in a certified run, which is what makes the accumulator append-only rather than a dictionary.

\begin{lemma}[Canonicity of insertion]
\label{lem:insert-canon}
For every finite partial map $M$, key $k \notin \mathrm{dom}(M)$ and value $v$,
\[
  \mathsf{ins}\bigl(\mathsf{Tree}(M);k,v\bigr) = \mathsf{Tree}\bigl(M \uplus \{k \mapsto v\}\bigr).
\]
\end{lemma}
\begin{proof}
Induction on $|\mathrm{dom}(M)|$, writing $T = \mathsf{Tree}(M)$. For $M$ empty, both sides are $\mathsf{Leaf}(k,v)$. Clause~2 never fires, because $T = \mathsf{Leaf}(k,\cdot)$ would put $k$ in $\mathrm{dom}(M)$.

Clause~3. The result is valid: the edges inside $T$ are unchanged, and the two new edges were checked above. Its leaves are those of $T$ together with $\mathsf{Leaf}(k,v)$, so it represents $M \uplus \{k \mapsto v\}$, and Proposition~\ref{prop:unique} identifies it as $\mathsf{Tree}(M \uplus \{k \mapsto v\})$.

Clause~4. By the construction in Proposition~\ref{prop:unique}, $p$ is the longest common prefix of $\mathrm{dom}(M)$, and $l = \mathsf{Tree}(M_{p\|0})$, $r = \mathsf{Tree}(M_{p\|1})$. Take $k[d]=0$; the other side is symmetric. Since $p \preceq k$, adding $k$ leaves the longest common prefix at $p$ and adds the binding to the left cone, so
\[
  \mathsf{Tree}\bigl(M \uplus \{k \mapsto v\}\bigr)
  = \mathsf{Node}\bigl(d,p,\mathsf{Tree}(M_{p\|0} \uplus \{k \mapsto v\}),r\bigr),
\]
again by the construction in Proposition~\ref{prop:unique}. The induction hypothesis applied to $M_{p\|0}$ makes the left child $\mathsf{ins}(l;k,v)$.
\end{proof}

Lemma~\ref{lem:insert-canon} is what makes insertion order irrelevant: iterated from the empty tree, $\mathsf{ins}$ yields $\mathsf{Tree}$ of the accumulated map regardless of the sequence, so the root commits to the recorded map and not to the history that produced it. The next definition computes the digest of the extended tree from a certificate alone, without access to $T$.

\begin{definition}[Update reading]
\label{def:update-read}
Let $C = \bigl((d_1,s_1),\ldots,(d_m,s_m);(k',v')\bigr)$ be a certificate (Definition~\ref{def:inclusion-cert}), let $k \in \mathcal{K}$ with $k \neq k'$, and let $v \in \mathcal{V}$. Put $\hat p = \mathsf{lcp}(k,k')$ and $\hat d = |\hat p|$, let $c_0,\ldots,c_m$ be the digests of Definition~\ref{def:inclusion-cert}, which depend only on the certificate, and let $t$ be the number of indices $j$ with $d_j > \hat d$; since the depths decrease, these are $j = 1,\ldots,t$. The \emph{derived certificate} replaces the openings deeper than $\hat d$ by a single opening at $\hat d$ whose sibling digest is $c_t$, and the terminal by the new binding:
\[
  C^{+} = \bigl((\hat d,\,c_t),\ (d_{t+1},s_{t+1}),\ldots,(d_m,s_m);\ (k,v)\bigr);
\]
for the empty certificate, $C^{+} = \bigl(\,;(k,v)\bigr)$. The \emph{update reading} $\mathsf{next}(C;k,v)$ is the unique root against which $C^{+}$ verifies for the target key $k$. If some $d_j$ equals $\hat d$, then $C^{+}$ is malformed and the update reading is undefined.
\end{definition}

The bit requirements of Definition~\ref{def:inclusion-cert} are vacuous for $C^{+}$---its terminal key is the target key---so $C^{+}$ verifies against exactly one root, the final digest of its fold; $\mathsf{next}$ itself compares nothing against any root. The undefined cases are exactly the certificates that cannot participate in an insertion of $k$: a terminal key $k' = k$ certifies that $k$ is already recorded, and an opening at depth $\hat d$ violates the requirement $k[\hat d] = k'[\hat d]$, so such a $C$ verifies for $k$ against no root. The two folds share their recursion, so a verifier computes the digests of $C$ and $C^{+}$ in one pass. Theorem~\ref{thm:update} is what licenses reading the new root off $C^{+}$ once verification of $C$ against the previous root has succeeded.

\begin{theorem}[Single-insertion update]
\label{thm:update}
Let $M$ be a finite partial map and $T = \mathsf{Tree}(M)$, let $k \in \mathcal{K}$, and let a certificate $C$ with terminal key $k' \neq k$, or the empty certificate, verify for $k$ against $\mathsf{dig}(T)$. Then $k \notin \mathrm{dom}(M)$ and, for every $v \in \mathcal{V}$,
\[
  \mathsf{next}(C;k,v) = \mathsf{dig}\bigl(\mathsf{Tree}(M \uplus \{k \mapsto v\})\bigr).
\]
\end{theorem}
\begin{proof}
The empty certificate verifies only against $r = \varnothing$, so $T = \varepsilon$ and $M = \emptyset$; here $C^{+}$ is the terminal-only certificate $(\,;(k,v))$, which verifies against $\mathsf{LeafHash}(k,v) = \mathsf{dig}(\mathsf{Tree}(\{k \mapsto v\}))$.

Let $C$ be nonempty. Theorem~\ref{thm:inclusion-cert}(i) gives $k \notin \mathrm{dom}(M)$, and its proof, through Lemma~\ref{lem:match}, identifies the openings of $C$ with the ancestors of $\mathsf{Leaf}(k',v')$ in $T$: the opening at depth $d_j$ coincides with the ancestor junction at that depth, $c_j$ is that junction's digest, $p_j = k'[0..d_j)$ is its true region (Lemma~\ref{lem:shape}(ii)), and $s_j$ is the digest of its subtree on the side away from $k'$. The openings list \emph{every} ancestor, because each opening's on-path child digest is the recomputed $c_{j-1}$, making consecutive openings parent and child in $T$.

No opening has depth $\hat d$: verification requires $k[d_j] = k'[d_j]$, and $k[\hat d] \neq k'[\hat d]$. Hence $C^{+}$ is well formed, the ancestors at depths $d_1,\ldots,d_t$ have regions containing the position $\hat d$, so none of them is a prefix of $k$, and the ancestors at depths $d_{t+1},\ldots,d_m$ have regions $p_j \preceq \hat p \preceq k$.

Now run $\mathsf{ins}(T;k,v)$. At the ancestors of depths $d_m,\ldots,d_{t+1}$, taken root downward, clause~4 of Definition~\ref{def:insert} applies and descends on side $k[d_j] = k'[d_j]$, toward $\mathsf{Leaf}(k',v')$. The next node $U$ on that path is the ancestor at depth $d_t$, or the terminal leaf itself if $t = 0$; its region agrees with $k'$ on at least the first $\hat d + 1$ bits, so $\mathsf{lcp}(\mathsf{reg}(U),k) = \mathsf{lcp}(k',k) = \hat p$, and clause~3 grafts $\mathsf{Node}(\hat d,\hat p,\cdot,\cdot)$ over $\mathsf{Leaf}(k,v)$ and $U$, with the new leaf on side $k[\hat d]$. By Lemma~\ref{lem:insert-canon}, $\mathsf{ins}(T;k,v) = \mathsf{Tree}(M \uplus \{k \mapsto v\})$.

In $\mathsf{ins}(T;k,v)$, the ancestors of $\mathsf{Leaf}(k,v)$ are the graft junction at depth $\hat d$, whose sibling subtree is $U$ with $\mathsf{dig}(U) = c_t$, and above it the junctions at depths $d_{t+1},\ldots,d_m$, whose sibling subtrees are unchanged with digests $s_j$. The certificate that Theorem~\ref{thm:inclusion-cert}(ii) constructs from this path is therefore exactly $C^{+}$, and it verifies against $\mathsf{dig}(\mathsf{ins}(T;k,v))$. Since $C^{+}$ verifies against a unique root, $\mathsf{next}(C;k,v) = \mathsf{dig}(\mathsf{ins}(T;k,v))$.
\end{proof}

\begin{corollary}[Derived inclusion certificate]
\label{cor:derived-cert}
In the setting of Theorem~\ref{thm:update}, $C^{+}$ is a verifying inclusion certificate for $(k,v)$ against the new root $\mathsf{next}(C;k,v)$---the certificate that Theorem~\ref{thm:inclusion-cert}(ii) constructs from $\mathsf{Tree}(M \uplus \{k \mapsto v\})$. The operator can therefore serve the inclusion certificate for a freshly recorded binding by this local transformation of the consumed non-inclusion certificate, without access to the tree.
\end{corollary}
\begin{proof}
This is the identification made in the closing paragraph of the proof of Theorem~\ref{thm:update}.
\end{proof}

\begin{definition}[Certified single-insertion run]
\label{def:run}
A \emph{certified run} is a sequence $(k_1,v_1,C_1,r_1),\allowbreak\ldots,\allowbreak(k_n,v_n,C_n,r_n)$ with $r_0 = \varnothing$ such that, for every $i$, the certificate $C_i$ verifies for $k_i$ against $r_{i-1}$, its terminal key differs from $k_i$ (or $C_i$ is empty), and $r_i = \mathsf{next}(C_i;k_i,v_i)$.
\end{definition}

\begin{theorem}[Run soundness]
\label{thm:run}
In a certified run the keys $k_1,\ldots,k_n$ are pairwise distinct and, with $M_i = \{k_1 \mapsto v_1,\ldots,k_i \mapsto v_i\}$,
\[
  r_i = \mathsf{dig}\bigl(\mathsf{Tree}(M_i)\bigr) \qquad \text{for } i = 0,\ldots,n.
\]
\end{theorem}
\begin{proof}
Induction on $i$; the base is $r_0 = \varnothing = \mathsf{dig}(\mathsf{Tree}(\emptyset))$. If $r_{i-1} = \mathsf{dig}(\mathsf{Tree}(M_{i-1}))$, then Theorem~\ref{thm:update} applied to $C_i$ gives $k_i \notin \mathrm{dom}(M_{i-1})$---so the keys remain distinct and $M_i$ is a partial map---and $r_i = \mathsf{dig}(\mathsf{Tree}(M_i))$.
\end{proof}

The witness maps $M_i$ are computed from the accepted transcript itself, which discharges clause~(i) of Definition~\ref{def:append-only-accumulator} in its literal form; nothing is assumed about what the operator knows. The protocol is complete as well as sound: for $k \notin \mathrm{dom}(M_{i-1})$ the operator produces a verifying certificate by the descent construction of Theorem~\ref{thm:inclusion-cert}(ii), and Theorem~\ref{thm:update} guarantees that the root it leads to is the canonical one.

\begin{corollary}[The run is append-only]
\label{cor:run-append}
In a certified run, no key is recorded twice and no recorded value changes: for $j \le i$, the tree committed by $r_i$ assigns $v_j$ to $k_j$, and the overwrite clause of Definition~\ref{def:insert} is never invoked. Moreover, a certificate for a recorded key $k_j$ whose terminal key differs from $k_j$ and which verifies against $r_i$ yields a collision.
\end{corollary}
\begin{proof}
By Theorem~\ref{thm:run}, $r_i$ commits $\mathsf{Tree}(M_i)$ with $M_i \supseteq M_j$ and distinct keys, so $M_i(k_j) = v_j$. The overwrite clause fires only for a key already in the tree, which Theorem~\ref{thm:update} excludes in every round. Finally, Theorem~\ref{thm:inclusion-cert}(i) applied to such a certificate concludes $k_j \notin \mathrm{dom}(M_i)$, contradicting $k_j \in \mathrm{dom}(M_j) \subseteq \mathrm{dom}(M_i)$; so the reduction behind it produces a collision.
\end{proof}

\begin{remark}[Derived, not checked]
\label{rem:derived-vs-checked}
In this protocol the post-state is a function of the pre-state: it is $\mathsf{ins}(T;k,v)$, and the certificate contributes no node to it. Nothing is reconstructed from data chosen by the untrusted operator, so the development needs no validity predicate over anything the operator sends: validity of every state tree is a theorem (Theorem~\ref{thm:run} exhibits the states as $\mathsf{Tree}(M_i)$), and the junction regions used in verification are derived from the terminal key, their correctness being a conclusion of Lemma~\ref{lem:match} rather than a check. The same holds for the user-facing certificates themselves: inclusion and non-inclusion verification (Definition~\ref{def:inclusion-cert}) involves no region locking. Indeed, no statement in this subsection or in Theorem~\ref{thm:inclusion-cert} uses the region operand of the junction hash---all of them would remain valid over a junction hash committing the depth alone, with Lemma~\ref{lem:match} then pinning depths and child digests only and the regions recovered from the matched tree through Lemma~\ref{lem:shape}(ii). The region commitment is needed exactly where tree shape is imported from the untrusted aggregator: the batch transcript attaches opaque preserved subtrees at prover-claimed positions, and there coherent placement must be \emph{checked}, which is what the region commitment makes locally checkable (Remark~\ref{rem:placement}).
\end{remark}

\begin{remark}[From one key to a batch]
\label{rem:single-key-border}
A round inserting $j$ bindings can be certified by $j$ applications of Theorem~\ref{thm:update} against $j$ intermediate roots, so the single-insertion theory is already a complete security argument for the accumulator. What it does not provide is a \emph{compressed} round: its cost is $j$ independent root-to-leaf certificates, whereas a transcript of the touched part of the tree shares the traversed structure across the batch and is verified in one pass. That compression is what forces the proof format of Section~\ref{sec:stack-verifier}, in which the post-state is assembled from prover-supplied opcodes rather than obtained as $\mathsf{ins}$ of the pre-state---and with it, the placement conclusions of Remark~\ref{rem:derived-vs-checked} become the checks of Algorithm~\ref{alg:stackverify}: edge coherence, confinement, and the region commitment that supports them. The next subsection proves that these checks suffice.
\end{remark}

\subsection{Certified Update Security}
\label{sec:consistency-theorem}

We now turn from certificates against one root---and runs of single insertions---to the batch verifier of Section~\ref{sec:stack-verifier} and the sequence of roots it accepts.

\begin{definition}[Certified history; append-only consistency]
\label{def:aoc}
A \emph{certified history} is a sequence $(B_1, \pi_1, r_1), \allowbreak \ldots, \allowbreak (B_n, \pi_n, r_n)$ with $r_0 = \varnothing$ and Algorithm~\ref{alg:stackverify} accepting $(\pi_i, r_{i-1}, r_i, B_i)$ for every $i$. Define cumulative maps by
\[
  M_0 = \emptyset, \qquad M_i = M_{i-1} \uplus B_i \quad (i=1,\ldots,n),
\]
identifying a batch with its induced partial map. The history is \emph{append-only consistent} if every disjoint union above is defined and
\[
  r_i = \mathsf{dig}(\mathsf{Tree}(M_i)) \qquad \text{for } i=0,\ldots,n.
\]
Thus definedness requires every batch to have pairwise distinct keys absent from all earlier batches. By Proposition~\ref{prop:unique}, this definition is equivalent to the existence of valid state trees with the previous round-by-round map-extension property.
\end{definition}

Definition~\ref{def:aoc} is the formal counterpart of clause~(i) in Definition~\ref{def:append-only-accumulator}; Theorem~\ref{thm:inclusion-cert} gives clause~(ii) for every canonical state tree. Theorem~\ref{thm:history} links those trees to the roots of a certified history. Definition~\ref{def:aoc} also subsumes prior-state preservation~\eqref{eq:psp}, with the committed maps of Section~\ref{sec:scope} realized as the cumulative maps $M_i$.

The verifier's checks are node-local. We make this explicit by giving the proof stream a typed syntax tree with recursively defined attributes; the checks become a local predicate on that tree. All reasoning below is structural induction. The algorithm itself appears only in Lemma~\ref{lem:parse}.

\begin{definition}[Proof terms]
\label{def:pterm}
Proof terms are generated by
\begin{align*}
P ::={} & S(c) \mid O(d, p, c_l, c_r) \mid O_L(k, v) \mid{}\\
        & L(k, v) \mid N(d, P_l, P_r)
\end{align*}
with operands in their domains: $c, c_l, c_r \in \{0,1\}^{\lambda}$, $k \in \mathcal{K}$, $v \in \mathcal{V}$, $0 \le d < \kappa$, $p \in \{0,1\}^{d}$. In $N(d, P_l, P_r)$, the subterm $P_l$ has side $\beta = 0$ and $P_r$ has side $\beta = 1$.

The \emph{advice} attribute is
\begin{align*}
\mathsf{adv}(S(c)) &= \bot, & \mathsf{adv}(O(d,p,\cdot,\cdot)) &= (d, p),\\
\mathsf{adv}(O_L(k,v)) &= (\kappa, k), & \mathsf{adv}(L(k,v)) &= (\kappa, k),
\end{align*}
and $\mathsf{adv}(N(d, P_l, P_r)) = (d, p)$ with the \emph{derived region} $p = \varrho_x[0..d)$, computed from any child with $\mathsf{adv}(P_x) = (\delta_x, \varrho_x) \neq \bot$. (The predicate $\mathsf{ok}$ below makes this well-defined.)

The \emph{old digest} attribute is
\begin{align*}
\mathsf{old}(S(c)) &= c, &
\mathsf{old}(O(d,p,c_l,c_r)) &= \mathsf{NodeHash}(d,p,c_l,c_r),\\
\mathsf{old}(O_L(k,v)) &= \mathsf{LeafHash}(k,v), &
\mathsf{old}(L(k,v)) &= \varnothing,
\end{align*}
and, for a junction $N = N(d, P_l, P_r)$ with derived region $p$, writing $o_x = \mathsf{old}(P_x)$,
\[
\mathsf{old}(N) = \mathsf{HashPreState}(d,p,o_l,o_r),
\]
the \emph{four-way rule} of Equation~\eqref{eq:hash-pre-state}. The \emph{new digest} attribute equals $\mathsf{old}$ on $S$, $O$, and $O_L$; further, writing $n_x = \mathsf{new}(P_x)$,
\begin{align*}
\mathsf{new}(L(k,v)) &= \mathsf{LeafHash}(k,v),\\
\mathsf{new}(N(d,P_l,P_r)) &= \mathsf{NodeHash}(d,p,n_l,n_r).
\end{align*}

The predicate $\mathsf{ok}(P)$ holds if every junction $N(d, P_l, P_r)$ in $P$ satisfies:
\begin{enumerate}[nosep,label=(\alph*)]
  \item \emph{coherence}: every advised child, $\mathsf{adv}(P_x) = (\delta_x, \varrho_x)$, has $\delta_x > d$ and $\varrho_x[d] = \beta_x$;
  \item \emph{agreement}: at least one child is advised, and all advised children yield the same $\varrho_x[0..d)$;
  \item \emph{confinement}: if $\mathsf{old}(P_l) = \varnothing$ or $\mathsf{old}(P_r) = \varnothing$, then both children are advised.
\end{enumerate}
$P$ \emph{accepts} $(r_{\mathsf o}, r_{\mathsf n}, B)$ if $\mathsf{ok}(P)$, $\mathsf{old}(P) = r_{\mathsf o}$, $\mathsf{new}(P) = r_{\mathsf n}$, and the labels of the $L$ leaves of $P$, read left to right, are exactly the elements of $B$ sorted, with strictly increasing keys.
\end{definition}

\begin{lemma}[Stream--term correspondence]
\label{lem:parse}
Let $B \neq [\,]$. Algorithm~\ref{alg:stackverify} accepts $(\pi, r_{i-1}, r_i, B)$ iff $\pi$ is the post-order serialization of a proof term that accepts $(r_{i-1}, r_i, B)$. The term is unique and the translation is linear-time.
\end{lemma}
\begin{proof}
Each of $S$, $O$, $O_L$, $L$ pushes one stack entry; $N$ pops two and pushes one. A run without underflow that ends with one entry parses the stream uniquely as the post-order serialization of a term; conversely, the serialization of any term replays without underflow and ends with one entry. By induction over the run, the triple pushed for a subterm is exactly $(\mathsf{old}, \mathsf{new}, \mathsf{adv})$ of that subterm, and the checks executed per opcode are exactly the operand-domain conditions of Definition~\ref{def:pterm} plus, at $N$, clauses (a)--(c) of $\mathsf{ok}$. The final comparisons equate the root digest pair with $(r_{i-1}, r_i)$; the batch consumption at $L$ and the final $b = |B|$ check equate the left-to-right $L$ labels with the sorted batch.
\end{proof}

\begin{lemma}[Old projection]
\label{lem:proj}
For a proof term $P$, define the \emph{old projection} $P_0$ by structural recursion: $S(c) \mapsto \mathsf{Hole}(c)$; $O(d,p,c_l,c_r) \mapsto \mathsf{ONode}(d,p,\mathsf{Hole}(c_l),\mathsf{Hole}(c_r))$; $O_L(k,v) \mapsto \mathsf{OLeaf}(k,v)$; $L(k,v) \mapsto$ nothing; $N(d,P_l,P_r)$ with derived region $p$: nothing if both children project to nothing, the surviving child's projection if exactly one does, and $\mathsf{ONode}(d,p,\cdot,\cdot)$ over the two projections otherwise. Then $P_0$ is empty iff $\mathsf{old}(P) = \varnothing$, and otherwise $P_0$ is an opening tree with $\mathsf{eval}(P_0) = \mathsf{old}(P)$.
\end{lemma}
\begin{proof}
Induction on $P$. A subterm projects to nothing iff its old digest is $\varnothing$: immediate for the terminals, and for $N$ the three projection cases match the four-way rule ($\varnothing$, the two pass-throughs, hashed). In the surviving cases, $\mathsf{eval}$ of the projection recomputes exactly the $\mathsf{old}$ value; the pass-through case forwards both the projection and the digest of the surviving child.
\end{proof}

\begin{lemma}[Reconstruction]
\label{lem:graft}
Let $T$ be a valid tree and let the proof term $P$ accept $(\mathsf{dig}(T), r_{\mathsf n}, B)$. Then there is a valid tree $T'$ with $\mathsf{dig}(T') = r_{\mathsf n}$ and $\mathsf{map}(T') = \mathsf{map}(T) \uplus B$.
\end{lemma}
\begin{proof}
First obtain a subtree of $T$ for every digest operand of $P$. If $\mathsf{dig}(T) = \varnothing$ then $T = \varepsilon$, and by Lemma~\ref{lem:proj} $P_0$ is empty; since a non-$\varnothing$ old digest survives every case of the four-way rule, $P$ then contains no $S$, $O$, or $O_L$, and no assignment is needed. Otherwise $\mathsf{eval}(P_0) = \mathsf{dig}(T)$, and Lemma~\ref{lem:match} matches $P_0$ into $T$: every $O$ and $O_L$ of $P$ coincides with the node of $T$ at its position, and every hole---an $S$ operand or a child digest of an $O$---is assigned the subtree of $T$ at its position.

Define $T' = \mathsf{graft}(P)$ by structural recursion: $S(c) \mapsto$ its assigned subtree; $O(d,p,\cdot,\cdot) \mapsto \mathsf{Node}(d,p,\cdot,\cdot)$ over its two assigned subtrees; $O_L(k,v) \mapsto \mathsf{Leaf}(k,v)$; $L(k,v) \mapsto \mathsf{Leaf}(k,v)$; $N(d,P_l,P_r) \mapsto \mathsf{Node}(d,p,\cdot,\cdot)$ with its derived region $p$, over the grafts of its children.

\emph{Digests.} By induction, $\mathsf{dig}(\mathsf{graft}(Q)) = \mathsf{new}(Q)$ for every subterm $Q$: $\mathsf{new}$ agrees with $\mathsf{dig}$ on the terminals and recomputes every junction. Hence $\mathsf{dig}(T') = r_{\mathsf n}$.

\emph{Advice.} Wherever $\mathsf{adv}(Q) = (\delta, \varrho) \neq \bot$, the top node of $\mathsf{graft}(Q)$ has depth $\delta$ and region $\varrho$. This holds by construction for $O$, $O_L$, $L$, and $N$.

\emph{Validity.} Every subterm grafts to a nonempty tree, so every junction of $T'$ has two nonempty children. Classify the edges of $T'$.
(1)~Edges inside assigned subtrees are edges of $T$.
(2)~The two edges below an $O$ graft: by the matching, the $O$ coincides with a junction of $T$ and its assigned children are the children of that junction in $T$; these are edges of $T$.
(3)~An edge whose child is an $S$ graft: the parent is an $N$; the $S$ child is unadvised, so by confinement both old digests at the parent are non-$\varnothing$, and the parent appears in $P_0$ as an $\mathsf{ONode}$ with its derived region. The matching places this junction, and the $T$-subtree assigned to the $S$ hole below it, in $T$; this is an edge of $T$.
(4)~Every remaining edge has an advised child. Coherence gives $\delta > d$ and, with agreement, $p\|\beta \preceq \varrho$; by the advice claim, $(\delta, \varrho)$ is the true depth and region of the child's top node. Edges of $T$ satisfy the validity conditions because $T$ is valid; edges of kind (4) satisfy them directly. Hence $T'$ is valid.

\emph{Map.} The leaves of $T'$ are the leaves of assigned subtrees, the $O_L$ leaves, and the $L$ leaves. By clause (d) of the matching, the first two groups are exactly the leaves of $T$. The $L$ leaves carry exactly the elements of $B$. By Lemma~\ref{lem:shape}(i), all leaf keys of the valid $T'$ are distinct; hence the keys of $B$ are absent from $\mathsf{map}(T)$, and $\mathsf{map}(T') = \mathsf{map}(T) \uplus B$.
\end{proof}

\begin{theorem}[Round soundness]
\label{thm:round}
Let $T$ be a valid tree with $\mathsf{dig}(T) = r_{i-1}$, and let Algorithm~\ref{alg:stackverify} accept $(\pi, r_{i-1}, r_i, B)$. Then there is a valid tree $T'$ with $\mathsf{dig}(T') = r_i$ and $\mathsf{map}(T') = \mathsf{map}(T) \uplus B$. In particular, the keys of $B$ are fresh: no accepting run exists for a batch that re-records a present key, short of a collision.
\end{theorem}
\begin{proof}
If $B = [\,]$, acceptance forces $\pi = [\,]$ and $r_i = r_{i-1}$; take $T' = T$. Otherwise Lemma~\ref{lem:parse} yields a proof term accepting $(r_{i-1}, r_i, B)$, and Lemma~\ref{lem:graft} yields $T'$.
\end{proof}

\begin{definition}[Difference generator]
\label{def:generator}
Let $T$ be valid, let $B \neq [\,]$ have pairwise distinct keys absent from $M=\mathsf{map}(T)$, and put $T'=T\oplus B$. For a subtree $U$ of $T'$, let
\[
  M[U] = \{\,k \mapsto v \in M \mid \mathsf{Leaf}(k,v)\text{ occurs below }U\,\}.
\]
The difference generator traverses $T'$ from the root. On reaching a maximal subtree $U$ containing no key of $B$, it stops: Lemma~\ref{lem:persist}(iii) identifies $U$ as a preserved subtree of $T$. It emits $S(\mathsf{dig}(U))$ when the parent junction also occurs in $T$; when the parent is new, it emits the opened root of $U$, namely $O$ with the true node labels and child digests or $O_L$ with the true leaf. A leaf from $B$ emits $L(k,v)$. Every other node $U=\mathsf{Node}(d,p,U_l,U_r)$ emits $N(d,Q_l,Q_r)$ after recursively generating $Q_l$ and $Q_r$. The resulting proof term is $P_B$.

Operationally, $T'$ and $P_B$ are produced together by merging the sorted batch into $T$. Preserved subtrees are emitted at their roots without traversal, so the merge performs constant work per emitted constructor and per batch element.
\end{definition}

\begin{lemma}[Generator invariant]
\label{lem:generator}
For every subterm $Q$ generated for a subtree $U$ of $T'$:
\begin{enumerate}[nosep,label=(\roman*)]
  \item $\mathsf{ok}(Q)$;
  \item $\mathsf{old}(Q)=\mathsf{dig}(\mathsf{Tree}(M[U]))$;
  \item $\mathsf{new}(Q)=\mathsf{dig}(U)$;
  \item either $Q=S(\mathsf{dig}(U))$ and $\mathsf{adv}(Q)=\bot$, or $\mathsf{adv}(Q)=(\mathsf{dep}(U),\mathsf{reg}(U))$;
  \item the $L$ labels of $Q$, from left to right, are exactly the bindings of $B$ below $U$, in strictly increasing key order.
\end{enumerate}
Consequently, $P_B$ accepts $(\mathsf{dig}(T),\mathsf{dig}(T'),B)$ and is computable, together with $T'$, in time $O(|P_B|+|B|)$ after sorting $B$.
\end{lemma}
\begin{proof}
Induction over the generator.

If $U$ is preserved, then $M[U]=\mathsf{map}(U)$ and Proposition~\ref{prop:unique} gives $U=\mathsf{Tree}(M[U])$. Both $S$ and the opened-root forms have equal old and new digest $\mathsf{dig}(U)$; the opened forms carry the true top-node advice, while $S$ is unadvised. There are no $L$ labels. If $U=\mathsf{Leaf}(k,v)$ comes from $B$, freshness gives $M[U]=\emptyset$; the emitted $L(k,v)$ has old digest $\varnothing=\mathsf{dig}(\mathsf{Tree}(\emptyset))$, new digest $\mathsf{dig}(U)$, true leaf advice, and the required singleton label list. Thus all five claims hold in the terminal cases.

Let $U=\mathsf{Node}(d,p,U_l,U_r)$ emit $Q=N(d,Q_l,Q_r)$ and apply the induction hypothesis to its children. At least one child term is not $S$: otherwise both child subtrees would contain no batch key and the generator would have stopped at $U$. Every advised child carries its true top-node depth and region, so validity of $U$ gives coherence; all advised children yield the true prefix $p$, so agreement holds and $\mathsf{adv}(Q)=(d,p)$.

The cone identity (Lemma~\ref{lem:shape}(iv)), applied to $U$ in $T'=\mathsf{Tree}(M\uplus B)$, says that $U$ contains exactly the bindings below $p$. Key placement at $U$ then gives $M[U_l]=M_{p\|0}$ and $M[U_r]=M_{p\|1}$; these maps partition $M[U]$. If both parts are empty, so is $M[U]$ and the four-way rule returns $\varnothing$. If exactly one part is nonempty, its unique valid tree is also $\mathsf{Tree}(M[U])$, and the rule passes through its digest. If both are nonempty, the two parts contain keys on opposite sides of $p$; the construction in Proposition~\ref{prop:unique} gives
\[
  \mathsf{Tree}(M[U])
  = \mathsf{Node}(d,p,\mathsf{Tree}(M[U_l]),\mathsf{Tree}(M[U_r])),
\]
so the hashed case gives its digest. This proves (ii).

If either child old digest is $\varnothing$, its old map is empty because no nonempty tree has digest $\varnothing$. Lemma~\ref{lem:persist}(i) then shows that the junction at $p$ is new. A preserved child below it was emitted opened, and every non-preserved child is advised by the induction hypothesis; hence both children are advised and confinement holds. Together with coherence and agreement this proves (i). The definition of $\mathsf{new}$ and the induction hypothesis give (iii). The preceding advice argument gives (iv). Finally, the $L$ list of $Q$ is the concatenation of its two child lists; Lemma~\ref{lem:shape}(iii) places every left key before every right key, proving (v).

At the root, $M[T']=M$, so Proposition~\ref{prop:unique} gives $\mathsf{Tree}(M)=T$. Claims (i)--(v) therefore say exactly that $P_B$ accepts $(\mathsf{dig}(T),\mathsf{dig}(T'),B)$. The operational merge of Definition~\ref{def:generator} visits each emitted constructor and batch element once, giving the stated time bound.
\end{proof}

\begin{theorem}[Completeness]
\label{thm:complete}
Let $T$ be a valid tree and $B$ a batch whose keys are pairwise distinct and absent from $\mathsf{map}(T)$. Then a stream $\pi_B$ such that Algorithm~\ref{alg:stackverify} accepts $(\pi_B, \mathsf{dig}(T), \mathsf{dig}(T \oplus B), B)$ is computable from $T$ and $B$ in linear time after sorting $B$.
\end{theorem}
\begin{proof}
For $B=[\,]$, take $\pi_B=[\,]$. Otherwise Lemma~\ref{lem:generator} shows that the generated term $P_B$ accepts the required roots and batch. Serialize $P_B$ in post-order and apply Lemma~\ref{lem:parse}; its linear-time bound and the generator bound give the claim.
\end{proof}

\begin{theorem}[History soundness]
\label{thm:history}
Every certified history (Definition~\ref{def:aoc}) is append-only consistent, or a collision is computable from the transcript.
\end{theorem}
\begin{proof}
Induction on rounds. The base map is $M_0=\emptyset$ and $r_0=\varnothing=\mathsf{dig}(\mathsf{Tree}(M_0))$. Assume $M_{i-1}$ is defined and $r_{i-1}=\mathsf{dig}(\mathsf{Tree}(M_{i-1}))$. Apply Theorem~\ref{thm:round} to the valid tree $\mathsf{Tree}(M_{i-1})$ and the accepted round $i$. It yields a valid $T'$ with
\[
  \mathsf{dig}(T')=r_i,
  \qquad
  \mathsf{map}(T')=M_{i-1}\uplus B_i.
\]
Thus the disjoint union defining $M_i$ exists, and Proposition~\ref{prop:unique} gives $T'=\mathsf{Tree}(M_i)$. Hence $r_i=\mathsf{dig}(\mathsf{Tree}(M_i))$, completing the induction.
\end{proof}

\begin{corollary}[Unicity]
\label{cor:unicity}
In a certified history, verifying inclusion certificates for $(k, v)$ against $r_i$ and for $(k, v')$ against $r_j$ imply $v = v'$.
\end{corollary}
\begin{proof}
Apply Theorem~\ref{thm:history}; if it produces a collision, we are done. Otherwise, without loss of generality let $i\leq j$. Then $r_i=\mathsf{dig}(\mathsf{Tree}(M_i))$ and $r_j=\mathsf{dig}(\mathsf{Tree}(M_j))$. Theorem~\ref{thm:inclusion-cert}(i) gives $M_i(k)=v$ and $M_j(k)=v'$. Since $M_i\subseteq M_j$, the two values are equal.
\end{proof}

\begin{corollary}[No out-of-batch bindings]
\label{cor:batchbound}
In a certified history, if an inclusion certificate for $(k,v)$ verifies against $r_i$, then
\[
  (k,v)\in B_1\uplus\cdots\uplus B_i.
\]
Here the $B_j$ are the consistency proof's witness batches. The claim does not assert that every witness binding came from an authenticated user request.
\end{corollary}
\begin{proof}
Apply Theorem~\ref{thm:history}. Unless it produces a collision, $M_i=B_1\uplus\cdots\uplus B_i$ and $r_i=\mathsf{dig}(\mathsf{Tree}(M_i))$; Theorem~\ref{thm:inclusion-cert}(i) gives the result.
\end{proof}

\begin{corollary}[No false non-inclusion]
\label{cor:noninc}
In a certified history, let two certificates for the same target key $k$ verify against $r_i$ and $r_j$ with $j\geq i$, the first with terminal key $k$ and the second with terminal key $k'\neq k$. Then a collision is computable from the two certificates and the transcript.
\end{corollary}
\begin{proof}
Apply Theorem~\ref{thm:history}; if it produces a collision, we are done. Otherwise the certified roots are the digests of $\mathsf{Tree}(M_i)$ and $\mathsf{Tree}(M_j)$. Theorem~\ref{thm:inclusion-cert}(i) gives
\[
  k\in\mathrm{dom}(M_i)\subseteq\mathrm{dom}(M_j)
  \quad\text{and}\quad
  k\notin\mathrm{dom}(M_j),
\]
a contradiction. Thus one of the reductions produces a collision.
\end{proof}

Because the two relations use one certificate object, this says that the terminal key of an accepted certificate for $k$ is monotone along a certified history: once some accepted certificate for $k$ terminates at $k$, no later root admits one that terminates elsewhere, short of a collision.

\begin{corollary}[Certified inclusion/non-inclusion service]
\label{cor:service}
Let $r_i$ be a root of a certified history.
\begin{enumerate}[nosep,label=(\roman*)]
  \item Given $\mathsf{Tree}(M_i)$, one certificate for $k$ verifying against $r_i$ is constructed in time linear in its path length. Its terminal key is $k$ when $k\in\mathrm{dom}(M_i)$ and differs from $k$ otherwise, so a single walk of the tree certifies either relation.
  \item From any certificate that verifies against $r_i$, the statement it certifies about $M_i$ follows.
\end{enumerate}
\end{corollary}
\begin{proof}
Apply Theorem~\ref{thm:history}; unless it produces a collision, $r_i=\mathsf{dig}(\mathsf{Tree}(M_i))$. The claims are then Theorem~\ref{thm:inclusion-cert}(ii) and~(i) respectively, including its construction and time bounds.
\end{proof}

Clause~(ii) is stated over a certificate that has been supplied, not over the set of strings that would verify. The stronger reading---that no verifying certificate for a recorded key \emph{exists}---does not follow from collision resistance: such a certificate is an encoded collision, and the assumption is that collisions cannot be found, not that they are absent.

\begin{remark}
Theorem~\ref{thm:inclusion-cert} establishes soundness and completeness of inclusion and non-inclusion certificates, and Theorems~\ref{thm:update} and~\ref{thm:run} lift them to a complete security argument for the accumulator updated one insertion at a time. Theorems~\ref{thm:round} and~\ref{thm:complete} do the same for each batched transition, while Theorem~\ref{thm:history} and Corollaries~\ref{cor:unicity}, \ref{cor:batchbound}, \ref{cor:noninc}, and~\ref{cor:service} lift the results to certified histories. Together they establish the accumulator interface of Definition~\ref{def:append-only-accumulator}. The core arguments use structural induction over trees, opening trees, or proof terms, with recursively defined attributes and explicit collision extraction; certificate correctness reduces directly to Matching and canonical path construction. The consistency algorithm appears only in Lemma~\ref{lem:parse}; a re-arithmetization of that verifier, such as the AIR of Section~\ref{sec:custom-air-circuit}, only needs to be proven equivalent to Algorithm~\ref{alg:stackverify}, and the rest of the argument carries over unchanged.
\end{remark}

\section{Custom AIR Circuit}
\label{sec:custom-air-circuit}

By proving correct execution of the consistency-proof verifier inside a succinct proof system, the proof shipped to the Consensus Layer becomes nearly independent of the batch size. The semantic public transition is the root pair $(r_{i-1}, r_i)$ and an old-root absence flag; a versioned proof envelope additionally carries the fixed protocol identifier and a small vector of scalar trace-shape counts. The batch, opcode stream, openings, and intermediate values are witness data and do not appear as public inputs. The implementation uses the STARK for computational integrity and succinctness, not for confidentiality: its committed traces are not zero-knowledge masked. Proving cost is dominated by the hash function: bit-oriented hashes such as SHA-256 are expensive to arithmetize, while arithmetization-friendly permutations such as Poseidon2 operate directly on field elements. A hand-built AIR (Algebraic Intermediate Representation) circuit over a small prime field, with a transparent FRI-based polynomial commitment scheme, exploits both effects and avoids the integer-to-field translation overhead of general-purpose zkVMs.

We implemented the consistency-proof verification logic as a Plonky3~\cite{plonky3} batch STARK over the BabyBear~\cite{babybear} 31-bit prime field, using Poseidon2~\cite{cryptoeprint:2023/323} (width 16, $\alpha=7$) for leaf and junction hashing. The arithmetization and its canonical proof protocol are open-source~\cite{rsmtair}.

\subsection{Underlying Tree Variant}

The implementation proves consistency for the RSMT of Section~\ref{sec:stack-verifier}: a path-compressed Patricia trie over 256-bit keys. A leaf commits to an exact 32-byte key and an exact 32-byte value through a three-step additive Poseidon2 sponge; applications that need variable-length values first hash them to this fixed-width value domain. An internal junction is the unique point where two non-empty subtrees diverge, and commits to both its absolute bifurcation depth and its canonical, left-aligned region. Empty old subtrees are represented by an optional digest; at the public boundary an explicit presence bit distinguishes absence from a present all-zero digest. Path compression eliminates single-child chains while the region commitment preserves key-determined placement. The circuit implements the full append-only-consistency statement of Definition~\ref{def:aoc}: prior-state preservation, canonical leaves and openings, edge coherence, and confinement of newly created junctions.

\subsection{From Stack Verifier to AIR Arithmetization}

The consistency proof is the flat post-order opcode stream of Section~\ref{sec:stack-verifier}. The arithmetization assigns one Table~A row to each $S$, $O$, $O_L$, $L$, or $N$ opcode, but it does not materialize the verifier's stack. Instead, every non-root opcode row is consumed exactly once as a child of a later $N$ row, and its digest pair and advice $(\delta,\varrho)$ are packed to one LogUp~\cite{cryptoeprint:2022/1530} tuple. Helper rows return the result of a leaf, opening, or join to the matching opcode row. Thus a junction receives exactly the two stack entries that the verifier algorithm would pop. The last real\footnote{Here, ``real'' means the row is not a padding row} Table~A row is constrained to the public roots and absence flag.

Canonical input digits and all three leaf permutations are fused into $L$ table, thus there is no need for a separate input batch table. Joins and openings are separated into $J$ and $O$ because they need different witnesses and canonicality rules. The last step is justified by the post-order/coherence ordering lemma: at every junction, coherence places every left-subtree key below every right-subtree key, while post-order emits the complete left subtree before the right. Consequently the subsequence of $L$ rows is strictly increasing in Table~A order and represents the existential batch. The result is the as-implemented $A/B/L/J/O/R/P$ architecture described next.

\subsection{AIR Tables}

The arithmetization is split across seven AIR tables that share one batch-STARK commitment. Each table is padded independently to a power-of-two height. All per-round witness values are in main traces; preprocessing consists only of verifier-reconstructible row indices, real-row and permutation-mode masks, and fixed lookup tables derived from the protocol version and the public scalar shape.

\begin{description}[nosep]
    \item[Table A (opcode execution)] has one row per opcode and one-hot selector columns for $S$\slash$O$\slash$O_L$\slash$L$\slash$N$. A row carries the old and new digests, the old-digest absence bit, advice $(\delta,\varrho)$, and the start index of its post-order subtree. Local rules implement the base-opcode semantics; the last real row binds the execution to $(r_{i-1},r_i)$ and the public old-root absence flag.
    \item[Table B (Poseidon2)] is the sole evaluator of every Poseidon2 permutation requested by the other tables, using eight vector lanes per row. It records occurrences rather than distinct inputs: two logically separate equal calls occupy multiplicity two. Feed-forward calls expose the full 16-element output, while terminal calls expose only the eight-element digest used by the RSMT.
    \item[Table L (canonical leaves)] has one row for every new leaf $L$ and opened leaf $O_L$. Twenty-six radix-$2^{10}$ digits each for the key and value reconstruct the nine field limbs of an exact 32-byte string; the digits are range-checked, and the reconstructed limbs feed all three steps of the leaf sponge. The row returns the digest together with the key limbs to its Table~A row. The different old-state semantics of $L$ and $O_L$ are imposed in A, so L needs no ``kind'' column.
    \item[Table J (junction joins)] has one row per $N$. It consumes the two child tuples, proves their contiguous post-order locations, derives their common parent region, checks increasing depth and the left/right divergence bit, and enforces confinement whenever the junction is new. The old digest follows the four cases $00/01/10/11$ (empty, right passthrough, left passthrough, or old-child hash) of Formula~\ref{eq:hash-pre-state} (the four-way-rule); the new side always hashes both new children. One node-prefix permutation is shared locally between the old and new child blocks when both are needed.
    \item[Table O (canonical junction openings)] has one row per $O$. It reconstructs the opened region from range-checked digits, proves that it is exactly a left-aligned $d$-bit prefix with a zero suffix, hashes the opened child digests, and returns the unchanged digest and advice to A. Separating O from J both removes join-only columns from opening rows and closes the canonical-region gap of the union layout.
    \item[Table R (range)] is the fixed relation $\{(b,x):0\leq b\leq10,\ 0\leq x<2^b\}$, padded to 2048 rows. One global range bus serves leaf digits, region digits, depths, depth gaps, and the bounded pieces of the coherence equations; its main multiplicities are fixed by global bus balance.
    \item[Table P (powers)] is the fixed table $(e,2^e)$ for $0\leq e\leq30$, padded to 32 rows. It anchors the boundary-bit powers used when J and O split a packed region limb at an arbitrary depth.
\end{description}

Seven global LogUp relations tie the tables together:
\begin{enumerate}[nosep]
    \item \emph{tree}: every non-last real A row sends $(\mathsf{row},\mathsf{subtree\_start},\mathsf{old},\mathsf{new},\mathsf{none},\mathsf{has},\delta,\varrho)$; J receives these tuples as left and right children;
    \item \emph{parent}: J and O return their complete parent tuple to the matching A row;
    \item \emph{leaf}: L returns $(\mathsf{row},\mathsf{digest},\mathsf{key})$ to an $L$ or $O_L$ row of A;
    \item \emph{p2ff} and \emph{p2term}: B supplies respectively full-output feed-forward permutations and digest-only terminal permutations to L, J, and O;
    \item \emph{range}: R supplies every bounded digit or helper value used by L, J, and O; and
    \item \emph{pow2}: P supplies the powers used by the region-boundary equations in J and O.
\end{enumerate}

The tree relation and the $\mathsf{subtree\_start}$ equations make the A rows one acyclic, spanning, contiguous post-order tree. The leaf and parent relations bind helper results back to the correct opcode row, while digest and advice always travel in the same tuple. The two Poseidon2 relations bind every sponge step to permutation occurrence; the range and power relations turn the packed-region field equalities into the intended integer and bit-prefix equalities. Composing these relations with the local constraints yields an execution of the complete reference stack verifier on an existential sequence of canonical 32-byte leaves. Its final state is the public transition and satisfies append-only consistency (Definition~\ref{def:aoc}).

\subsection{Design Optimizations and Implementation Gotchas}

The main optimizations in the presented implementation are also responses to failure modes found while building the earlier layouts:

\begin{itemize}[nosep]
    \item \emph{Do without the stack and the batch table.} Row-indexed multiset relations replace an explicit stack/RAM table. The new-leaf ordering lemma replaces a verifier-supplied sorted-batch table: topology and coherence already force the $L$ keys into strict order. Accordingly, the implementation proves that \emph{some} canonical batch realizes the transition; attesting a particular external queue would additionally require a public batch commitment and an in-AIR binding to it.
    \item \emph{Fuse leaves, split unlike row kinds.} Fusing canonical leaf input with its three sponge calls removes cross-row continuity and a batch bus. Conversely, separating J and O avoids paying join width on openings and permits O to enforce the stronger canonical-region relation it actually needs.
    \item \emph{Derive cheap selectors and helpers.} A derives ``has advice'' from the opcode, and J derives old-state cases, gaps, and common-prefix pieces where doing so does not raise the constraint degree unacceptably. Materialized helpers remain only where they keep lookup tuples linear.
    \item \emph{Prove post-order topology algebraically.} A free left-child pointer plus multiset balance can leave a disconnected functional-graph cycle; excluding it would silently borrow a Poseidon2 fixed-point assumption. The $\mathsf{subtree\_start}$ recurrence makes child indices strictly decrease and forces the root's subtree to start at row zero, ruling out cycles without a cryptographic assumption.
    \item \emph{Count hash occurrences, not distinct inputs.} Global Poseidon2 deduplication is unsound for completeness when one sender of multiplicity one must satisfy two equal logical receives. The arithmetization records every occurrence and shares only the node-prefix call that is explicitly local to one J row. A future deduplicated setup would need explicit sender multiplicities and no-wrap checks.
    \item \emph{Separate terminal from feed-forward hash buses.} A feed-forward output is another permutation's full state, whereas a terminal output contributes only an eight-limb digest. Masking the unused tail inside one lookup tuple raised the tuple degree beyond what the pinned batch-STARK path supported, so we use separate \emph{p2ff} and \emph{p2term} relations. This also ensures the propagated digest column itself is bound to the permutation output.
    \item \emph{Keep digest and placement advice inseparable.} Sending them on different buses would allow a prover to splice one row's digest to another row's region. The tree and parent tuples therefore carry the digest, absence bit, depth, region, and subtree start together, keyed by a verifier-fixed row index.
    \item \emph{Make setup verifier-owned.} We put all witness data in main traces; the verifier validates a bounded scalar shape and reconstructs every AIR, mask, lookup definition, and preprocessing commitment independently. The protocol tag fixes the table layouts, encodings, hash constants, transcript, and FRI configuration.
    \item \emph{Treat padding and empty rounds as protocol cases.} Padding rows are constrained to a syntactic zero form and have zero lookup multiplicity; otherwise independently padded tables can create spurious bus traffic. A truly empty batch is not represented by a zero-row AIR with an unconstrained boundary, but by a separate identity-transition object that checks equality of present old and new roots.
    \item \emph{Gate lookup-column optimizations end to end.} Combining two same-direction global LogUp receives appeared to save extension columns, but failed the out-of-domain consistency check in the pinned Plonky3 implementation. Thus, we retain one lookup entry per context.
\end{itemize}

The implementation validates these choices with per-table constraint tests, cross-table balance tests, differential tests against the reference verifier, and adversarial trace mutations~\cite{rsmtair}. The negative cases include corrupted subtree starts, digest/advice splicing, non-canonical leaf or region digits, missing advice at a new junction, duplicated or omitted permutation occurrences, forged range or power multiplicities, nonzero padding, malformed shapes, and altered public roots.

\subsection{Cryptographic Setup}

The proof system uses BabyBear and its degree-four extension for the AIR and LogUp challenges. Its polynomial commitment scheme is \texttt{TwoAdicFriPcs} backed by Merkle trees, and Fiat--Shamir uses a duplex Poseidon2 challenger. The production protocol fixes the field, Poseidon2 constants, AIR and bus definitions, transcript domain, public-shape limits, and FRI parameters under one versioned identifier; none is prover-selectable. The presented configuration targets a conjectured $> 100$-bit standalone STARK/FRI soundness level. No trusted setup is required. The benchmark configurations in Section~\ref{sec:measured} are performance experiments.

\subsection{Measured Performance}
\label{sec:measured}

Hardware: 10-core Apple M1 Pro; release build utilizing the 8 performance cores; one consistency proof per row.

\begin{table}[ht]
\centering
\caption{Default configuration ($\sim$116-bit conjectured soundness, Poseidon2 as FRI hash, fresh-batch workload).}
\label{tab:perf-default}
\small
\begin{tabular}{r|rrrr|r}
\toprule
\textbf{batch} & \textbf{wit} & \textbf{trace} & \textbf{prove} & \textbf{verify} & \textbf{proof} \\
\midrule
1024 & 11\,ms & 5\,ms & 149\,ms & 44\,ms & 1.69\,MB \\
4096 & 20\,ms & 4\,ms & 143\,ms & 45\,ms & 1.76\,MB \\
8192 & 37\,ms & 5\,ms & 228\,ms & 47\,ms & 1.82\,MB \\
\bottomrule
\end{tabular}
\end{table}

\begin{table*}[htb]
\centering
\caption{Effect of FRI parameters (batch 4096).}
\label{tab:perf-knobs}
\small
\begin{tabular}{l|rrr}
\toprule
\textbf{configuration} & \textbf{prove} & \textbf{verify} & \textbf{proof} \\
\midrule
default & 143\,ms & 45\,ms & 1.76\,MB \\
small proof (blowup\,$=2$, q\,$=50$) & 198\,ms & 24\,ms & 0.93\,MB \\
high grinding (bits\,$=24$, q\,$=92$) & 985\,ms & 42\,ms & 1.62\,MB \\
\bottomrule
\end{tabular}
\end{table*}

\begin{table}[ht]
\centering
\caption{Choice of FRI / commitment hash (batch 4096, internal tree hash is always Poseidon2).}
\label{tab:perf-frihash}
\small
\begin{tabular}{l|rrr}
\toprule
\textbf{FRI hash} & \textbf{prove} & \textbf{verify} & \textbf{proof} \\
\midrule
Poseidon2 & 138\,ms & 41\,ms & 1.76\,MB \\
SHA-256   & 149\,ms & 18\,ms & 1.69\,MB \\
Blake3    & 130\,ms & 12\,ms & 1.69\,MB \\
\bottomrule
\end{tabular}
\end{table}

A few observations on the measured data:

\begin{itemize}[nosep]
    \item Sustained proving throughput on a single CPU is $\sim$28\,000 tx/s at the default config and $\sim$36\,000 tx/s at a batch of 8192. This comfortably exceeds the original design target of $10\,000$ tx/s per shard.
    \item Increasing the FRI blowup from $2^1$ to $2^2$ halves the proof size and the verifier work, at the cost of $\sim$1.4$\times$ prover time. This is the natural ``small proof'' operating point.
    \item Increasing the Fiat--Shamir grinding bits from $16$ to $24$ slows the prover by $\sim$5--7$\times$ for an $\sim$8\% proof-size reduction; not a useful trade in this parameter range.
    \item The choice of FRI / Merkle hash is largely independent of in-circuit performance. For a \emph{native} verifier (e.g., the Consensus Layer's BFT core), Blake3 gives a $3\times$ faster verification at the same prove time. For \emph{recursive} verification inside another circuit, Poseidon2 is the natural choice because it is field-friendly.
    \item Table B (Poseidon2 evaluations) dominates the cell count, accounting for $\sim$76\% of the trace across all batch sizes. Further proving-side gains will come primarily from improvements to the Poseidon2 AIR.
    \item Verification time is essentially flat ($\sim$40\,ms with Poseidon2 FRI, $\sim$12\,ms with Blake3) across the measured batch range, dominated by FRI Merkle openings rather than by anything that scales with batch size.
\end{itemize}

Determinism is preserved across serial and parallel builds: \texttt{prove\_batch} output is byte-identical between the two, for the same challenger seed.

\subsection{Deployment}

Operationally, the AIR-based consistency-proof producer is a drop-in component of the Unicity aggregator implementation. The proof's public statement is unchanged ($r_{i-1}, r_i$); only the proof bytes shipped from the Aggregation Layer to the Consensus Layer change. Speculative execution of the next round overlaps with the unicity-certificate wait of the current round, so the AIR's prover latency does not lengthen the user-visible round time.

The entire proving stack runs on a single CPU. Unlike L1 ZK-rollup proving pipelines, there is no GPU farm, no proving market, no proving service in the critical path of a round. A single Aggregator that sustains $10\,000$ insertions per second fits within the power budget of a laptop charger. The one place where heavier proving does occur---the recursive aggregation of the next section---runs off the critical path, at checkpoint cadence, and may be performed by any independent operator.

\section{Proof Aggregation and Full-History Audit}
\label{sec:aggregation-audit}

The per-round consistency proofs of Section~\ref{sec:custom-air-circuit} are verified by the BFT Core before certification; a user who relies on a Unicity Certificate therefore relies on a quorum of the committee having performed that verification honestly. This section removes that remaining reliance. All per-round proofs, from all shards, over the entire operating history of the system, are folded into a single fixed-size STARK. Verifying this one proof establishes---under cryptographic assumptions alone---that every certified state transition since genesis was consistent, reducing the BFT committee's residual role from ``trusted to compute correctly'' to ``accountable for not equivocating''.

\subsection{The Aggregate Statement}
\label{sec:aggregate-statement}

Fix a network instance identifier $\alpha$ and its genesis configuration digest $g$, which commits to the initial sharding scheme $\mathcal{SH}_0$, its empty shard roots, and the genesis Trust Base entry. In round $i$, let $\mathcal{SH}_i$ be the active prefix scheme, let $c_i=H(\mathcal{SH}_i)$ be its canonical commitment, and let $R_i$ be the root of the prefix-shaped shard-root tree over $\{r_{i,\sigma}\}_{\sigma\in\mathcal{SH}_i}$ (Section~\ref{sec:sharding-architecture}). Let $D_i$ be an append-only commitment---a Merkle Mountain Range~\cite{mmr} (MMR)---to the certified sequence $\big((1,c_1,R_1),\ldots,(i,c_i,R_i)\big)$. The MMR supports appending an element (recomputable in-circuit from $D_{i-1}$ and a logarithmic-size witness) and proving membership of any $(j,c_j,R_j)$ against $D_i$ with a logarithmic-size path.

\begin{definition}[Round correctness]
\label{def:roundok}
$\mathsf{RoundOK}(i)$ holds iff
\begin{enumerate}[nosep]
  \item for every shard that persists from $\mathcal{SH}_{i-1}$ to $\mathcal{SH}_i$, either its root is unchanged or its consistency proof $\pi_{i,\sigma}$ verifies the transition $(r_{i-1,\sigma},r_{i,\sigma})$ per Algorithm~\ref{alg:stackverify} (in its STARK-compressed form);
  \item if $\mathcal{SH}_i\ne\mathcal{SH}_{i-1}$, every change replaces a prefix $\sigma$ by $\sigma\|0$ and $\sigma\|1$, and an authenticated split witness proves that the two child seed roots select exactly the corresponding subtrees of $r_{i-1,\sigma}$; any same-round child updates are then covered by ordinary consistency proofs from those seed roots;
  \item $c_i=H(\mathcal{SH}_i)$ and $R_i$ is the root of the shard-root tree determined by $\mathcal{SH}_i$ and its current shard roots; and
  \item $D_i$ extends $D_{i-1}$ by exactly the element $(i,c_i,R_i)$.
\end{enumerate}
\end{definition}

The aggregate statement $\mathcal{A}_n$ is then: \emph{starting from the genesis configuration $g$, there exists a sequence of rounds $1, \ldots, n$ such that $\mathsf{RoundOK}(i)$ holds for every $i$}. Its public values are
\[ \mathsf{pv}_n = (\alpha,\, g,\, n,\, c_n,\, R_n,\, D_n). \]
Everything else---the intermediate sharding schemes and shard roots, per-round consistency and split proofs, and MMR witnesses---is non-public witness data, so the proof and its statement have constant size regardless of $n$ and of the number of shards.

Note that validator signatures are absent from the statement. The aggregate proof does not verify the committee's quorum signatures, because the truth of $\mathcal{A}_n$ does not depend on \emph{who} certified the transitions; consistency is a property of the data itself. The link to the certified reality is made by the verifier, outside the proof, by checking that the roots it cares about coincide with Unicity Certificates validated against the Trust Base (Section~\ref{sec:maxi-validation}). Folding signature verification into the circuit is possible but adds cost without adding security: it still could not prevent a quorum from signing two divergent histories, which remains the only residual attack (Section~\ref{sec:residual-trust}).

\subsection{Recursive Aggregation}

The statement has the classic shape of incrementally verifiable computation~\cite{valiant2008ivc}: a step relation ($\mathsf{RoundOK}$) iterated from a known initial state ($g$), with a succinct proof carried forward. The aggregate proof $\Pi_i$ attests:
\begin{enumerate}[nosep]
  \item there exists a valid aggregate proof $\Pi_{i-1}$ for public values $\mathsf{pv}_{i-1}$ under the same program (or $i = 1$ and $\mathsf{pv}_0$ is the genesis state derived from $g$); and
  \item $\mathsf{RoundOK}(i)$ holds with respect to $\mathsf{pv}_{i-1} \to \mathsf{pv}_i$.
\end{enumerate}
In practice, one recursion step folds a contiguous span of rounds rather than a single one; the prover chooses the span to balance proving latency against per-step overhead. The resulting proof chain is checkpointed at a protocol-defined cadence (e.g., hourly), and only the latest checkpoint needs to be retained or distributed.

\subsection{Instantiation on the SP1 zkVM}
\label{sec:sp1-instantiation}

We instantiate the recursion on the SP1 zkVM~\cite{sp1}, which supports the required proof composition natively: a guest program (ordinary Rust, compiled to RISC-V) can verify SP1 \emph{compressed} proofs of other guest programs---including itself---as an in-VM operation whose cost is small and independent of the size of the computation folded so far. The aggregation guest program does the following:

\begin{enumerate}[nosep]
  \item reads $\mathsf{pv}_{i-1}$ and verifies the previous aggregate proof against the aggregation program's own verifying-key digest, which is itself bound into the public values; the top-level verifier checks this digest once, against the value it knows from the software distribution, closing the recursion;
  \item for each round in the span and each shard with a changed root, verifies the shard's Plonky3 consistency STARK by running the STARK verifier as guest code. The shard proofs destined for aggregation use the SHA-256 FRI configuration (Table~\ref{tab:perf-frihash}) precisely so that this step is cheap: the verifier's work is dominated by FRI Merkle-path hashing, which maps directly onto the zkVM's SHA-256 precompile;
  \item when the sharding-scheme commitment changes, verifies the split witnesses that bind each child seed root to the last certified parent root;
  \item recomputes each round's scheme commitment $c_i$ and global root $R_i$, appends $(i,c_i,R_i)$ to the MMR, and commits $\mathsf{pv}_i$ as the new public values.
\end{enumerate}

The native verification cost of one shard proof is $\sim$18\,ms (Table~\ref{tab:perf-frihash}); executed in the zkVM, with precompile acceleration, this translates to on the order of $10^7$--$10^8$ RISC-V cycles, i.e., seconds of proving time per shard-round on server-class hardware. This is orders of magnitude more expensive than the shard's own proving, which is exactly why aggregation is kept off the critical path: it runs behind the certified tip at checkpoint cadence and never delays a round.

The output of each step is a compressed STARK of constant size, transparent (no trusted setup), with hash-based assumptions only. Where an external system wants to check the checkpoint cheaply (e.g., a bridge contract), the compressed STARK can additionally be wrapped into a Groth16 or PLONK proof of a few hundred bytes; the audit path described below does not depend on such a wrapping and retains the transparent STARK end to end.

\subsection{The Aggregation Prover}
\label{sec:aggregation-prover}

Producing $\Pi_n$ is a pure computation over public data (Section~\ref{sec:data-availability}): the statement is deterministic, requires no private inputs and no authorization, and its output is universally verifiable. No particular operator has a privileged role in its production.

Any operator with a mirror of the round archive can extend the latest checkpoint and publish $(\Pi_{n'},\mathsf{pv}_{n'})$ for $n'>n$. Acceptance is mechanical: the proof must verify, its public values must extend the recorded checkpoint, and the claimed final root must match the folded shard-root transitions. BFT Core validators are natural operators because they already retain the round artifacts, but they hold no privileged proving capability; an independent archive mirror produces the same proof. How deployments arrange the operation of this optional service is outside the protocol statement analyzed here. If no prover is available, immediate certificate-based operation is unaffected and only audit latency increases.

\subsection{Data Availability}
\label{sec:data-availability}

Two distinct data planes must remain available.

\emph{Serving plane.} The contents of each shard's tree---the recorded keys and values---are needed to serve inclusion and non-inclusion proofs to users. This state is replicated within the shard's own validator cluster (Section~\ref{sec:practical}); its loss is a liveness failure of that shard, not a safety failure of the system, since the certified roots and the append-only discipline persist independently.

\emph{Audit plane.} The per-round artifacts consumed by the aggregation prover: for each round $i$, the tuples $(r_{i-1,\sigma}, r_{i,\sigma}, \pi_{i,\sigma})$ for every changed shard, the active sharding scheme and any split witnesses, the global root $R_i$, the issued certificates, the Trust Base entries at configuration changes, and the accepted checkpoint proofs. All of these are public by design. The insertion batches are not public inputs to the \emph{shard} proofs and are not needed for aggregation, although the unmasked shard STARK does not itself claim witness confidentiality. The artifacts are published to a content-addressed round archive replicated by the shards and the BFT Core validators, which anyone may mirror. At $\sim$1--2\,MB per changed shard per round, archive growth is proportional to the number of shard transitions rather than transaction volume.

The aggregate proof gives the archive a bounded retention requirement: once a span of rounds has been folded into an accepted checkpoint (plus a safety margin for independent re-verification), its artifacts may be pruned, because $\Pi_n$ subsumes them. Symmetrically, the archive gives the prover role its permissionless character: a new prover bootstraps from the latest accepted checkpoint and the archive tail, with no handover, registration, or historical sync beyond the unpruned window. New provers can join permissionlessly.

\subsection{Maximalist Validation Procedure}
\label{sec:maxi-validation}

A user applying the maximalist model of Section~\ref{sec:maximalist} validates as follows.

\emph{Setup, once:} obtain the network instance parameters $(\alpha, g)$ and the aggregation program's verifying-key digest from the software distribution. This is the entire root of trust of the audit path.

\emph{Per audit:}
\begin{enumerate}[nosep]
  \item Obtain the latest checkpoint $(\Pi_n, \mathsf{pv}_n)$ from any source; the source need not be trusted. Verify the STARK $\Pi_n$ against the known verifying-key digest, and check $\mathsf{pv}_n.\alpha = \alpha$ and $\mathsf{pv}_n.g = g$. This takes milliseconds to seconds on commodity hardware.
  \item For each certificate in the received token's history that references a round $j \leq n$ with sharding-scheme commitment $c_j$ and global root $R_j$: verify the MMR membership of $(j,c_j,R_j)$ in $\mathsf{pv}_n.D_n$ (a logarithmic-size path, obtainable from the round archive). This confirms both that the key was routed under the authenticated scheme and that the root anchoring its inclusion proof lies on the proven, non-forking history.
  \item For certificates younger than the checkpoint ($j > n$): validate them against the Trust Base as in the practical model (Section~\ref{sec:practical}), and re-audit when the next checkpoint arrives.
\end{enumerate}

If the user ever encounters two verifying checkpoints, or a verifying checkpoint and a quorum-signed certificate, that assign different global roots to the same round, this pair is publishable evidence of committee equivocation, and the user rejects the affected token history.

\subsection{Residual Trust Analysis}
\label{sec:residual-trust}

The procedure above makes the reduction of assumptions precise. Under collision resistance of the hash function and soundness of the two proof systems (the shard STARKs and the aggregation zkVM), the following hold \emph{unconditionally}, with no assumption about the validator set:

\begin{enumerate}[nosep]
  \item every round transition in the proven history satisfied append-only consistency (Definition~\ref{def:aoc}, Theorem~\ref{thm:history})---no recorded token state was ever deleted, modified, or re-recorded, hence no double-spend is representable within that history; and
  \item the proven history is linear: each round extends its predecessor, and $D_n$ commits to the unique sequence of global roots.
\end{enumerate}

The honest-quorum assumption on the BFT Core is thereby removed from the correctness argument for the proven history. Two consensus-level properties remain: \emph{liveness} (rounds continue to be certified and proofs remain available for aggregation) and \emph{uniqueness of the tip} (a quorum could sign two divergent continuations, each internally consistent and each separately provable). The latter cannot be excluded by any proof system, since it is a statement about signing behavior rather than computation; however, two conflicting artifacts make it detectable and attributable to specific signing keys. In summary, correctness of the recorded history rests on cryptography, whereas immediate progress and a unique live tip retain the standard BFT quorum assumptions.

\section{Summary}

Succinct argument systems offer a powerful method for proving the execution of a computation, in our case, checking consistency proofs of a distributed cryptographic data structure. For use cases with small changesets, a simple hash-based proof, whose size is linear in the batch size, is optimal. However, as batch sizes increase and bandwidth becomes a constraint, the constant or near-constant size proofs generated by STARK systems become more advantageous.

A second axis of optimization is the scope of the in-circuit statement itself. Drawing that scope correctly requires care: the seemingly sufficient structure-only statement admits a concrete cross-round equivocation attack (Remark~\ref{rem:placement}). The right kernel---\emph{append-only consistency}, i.e., prior-state preservation plus coherent placement of insertions, made locally checkable by the region-committing node hash---is proved sufficient in Theorem~\ref{thm:history}, while completeness and service properties remain self-policed by the protocol around the public root. This still turns a verification problem that naively would require reconstructing the globally unique tree shape into one with a tight, hand-built AIR. The result is the measured proving throughput reported in Section~\ref{sec:measured}: a single Aggregator on a single CPU sustainably exceeds $10\,000$ insertions per second with a 1.7\,MB transparent proof and tens of milliseconds verification time. This shows that operating a very large-scale Aggregation Layer with cryptographically verified updates is practical today, with no GPU farm or trusted setup, and on a power budget compatible with commodity hardware.

Above the per-round proofs sits the aggregation mechanism of Section~\ref{sec:aggregation-audit}: the consistency proofs of all shards and the Consensus Layer's own state transitions are folded, by recursion on a zkVM, into a single fixed-size transparent proof of the correctness of the system's entire operating history. The two validation paths complement each other. The pragmatic path, anchored in the Unicity Trust Base, delivers finality within seconds under the standard BFT quorum assumption; the audit path, available with a delay of minutes to hours, retrospectively removes that assumption from the correctness of the recorded history. The consistency-proof mechanism now carries a formal correctness proof (Theorem~\ref{thm:history} and its corollaries), built in two tiers: a single-insertion theory in which a non-inclusion certificate alone drives each update and no placement check is needed at all (Section~\ref{sec:single-key}), and the batched transcript whose placement checks are proved to enforce the same map-level guarantee. In both tiers, append-only consistency of the entire certified history reduces to the collision resistance of the underlying hash function. A machine-checked formalization of this proof is the natural next step.

\bibliographystyle{plain}
\bibliography{aggregation-layer}

@misc{wp,
  author = {{The Unicity Developers}},
  title = {Unicity Whitepaper},
  year = {2025},
  publisher = {{GitHub}},
  journal = {{GitHub} repository},
  howpublished = {\url{https://github.com/unicitynetwork/whitepaper/releases/tag/latest}}
}

@misc{sp1,
  author = {{Succinct Labs}},
  title = {{SP1}},
  year = {2025},
  publisher = {{GitHub}},
  journal = {{GitHub} repository},
  howpublished = {\url{https://github.com/succinctlabs/sp1}}
}

@misc{mmr,
  author = {Peter Todd},
  title = {Merkle Mountain Ranges},
  year = {2012},
  howpublished = {{OpenTimestamps} documentation},
  url = {https://github.com/opentimestamps/opentimestamps-server/blob/master/doc/merkle-mountain-range.md}
}

@inproceedings{valiant2008ivc,
  author = {Paul Valiant},
  title = {Incrementally Verifiable Computation or Proofs of Knowledge Imply Time/Space Efficiency},
  booktitle = {Theory of Cryptography ({TCC} 2008)},
  series = {LNCS},
  volume = {4948},
  publisher = {Springer},
  year = {2008}
}

@misc{bitcoin,
  author = {Satoshi Nakamoto},
  title = {Bitcoin: A Peer-to-Peer Electronic Cash System},
  howpublished = {White paper},
  url = {http://www.bitcoin.org/bitcoin.pdf},
  year = 2009
}

@misc{rsmtair,
  author = {Risto Laanoja},
  title = {{rsmt-air}: {Plonky3} {AIR} for {RSMT} Consistency Proofs},
  year = {2026},
  publisher = {{GitHub}},
  journal = {{GitHub} repository},
  howpublished = {\url{https://github.com/ristik/rsmt-air}}
}

@misc{plonky3,
  author = {{Polygon Zero Team}},
  title = {{Plonky3}: A Toolkit for {SNARK} and {STARK} Backends},
  year = {2025},
  publisher = {{GitHub}},
  journal = {{GitHub} repository},
  howpublished = {\url{https://github.com/Plonky3/Plonky3}}
}

@misc{cryptoeprint:2023/323,
  author = {Lorenzo Grassi and Dmitry Khovratovich and Markus Schofnegger},
  title = {{Poseidon2}: A Faster Version of the {Poseidon} Hash Function},
  howpublished = {Cryptology {ePrint} Archive, Paper 2023/323},
  year = {2023},
  url = {https://eprint.iacr.org/2023/323}
}

@misc{cryptoeprint:2022/1530,
  author = {Ulrich Hab\"ock},
  title = {Multivariate Lookups Based on Logarithmic Derivatives},
  howpublished = {Cryptology {ePrint} Archive, Paper 2022/1530},
  year = {2022},
  url = {https://eprint.iacr.org/2022/1530}
}

@inproceedings{dahlberg2016smt,
  author = {Rasmus Dahlberg and Tobias Pulls and Roel Peeters},
  title = {Efficient Sparse {Merkle} Trees: Caching Strategies and Secure (Non-)Membership Proofs},
  booktitle = {21st Nordic Conference on Secure {IT} Systems ({NordSec} 2016)},
  series = {LNCS},
  volume = {10014},
  publisher = {Springer},
  year = {2016},
  url = {https://eprint.iacr.org/2016/683}
}

@misc{rfc6962,
  author = {Ben Laurie and Adam Langley and Emilia K\"asper},
  title = {Certificate Transparency},
  howpublished = {{RFC} 6962, {IETF}},
  year = {2013},
  url = {https://www.rfc-editor.org/rfc/rfc6962}
}

@inproceedings{zerocash,
  author = {Eli Ben-Sasson and Alessandro Chiesa and Christina Garman and Matthew Green and Ian Miers and Eran Tromer and Madars Virza},
  title = {Zerocash: Decentralized Anonymous Payments from {Bitcoin}},
  booktitle = {2014 {IEEE} Symposium on Security and Privacy},
  year = {2014},
  url = {https://eprint.iacr.org/2014/349}
}

@misc{tornado,
  author = {Alexey Pertsev and Roman Semenov and Roman Storm},
  title = {{Tornado Cash} Privacy Solution, Version 1.4},
  year = {2019},
  howpublished = {Whitepaper},
  url = {https://berkeley-defi.github.io/assets/material/Tornado\%20Cash\%20Whitepaper.pdf}
}

@misc{exemodel,
  author = {Ahto Buldas and Dirk Draheim and Mike Gault and Risto Laanoja and Vladimir Rogojin and Ahto Truu},
  title = {The {Unicity} Execution Layer},
  year = {2026},
  howpublished = {{arXiv} preprint {arXiv}:2606.02181},
  url = {https://arxiv.org/abs/2606.02181}
}

@misc{predicates,
  author = {Ahto Buldas and Dirk Draheim and Mike Gault and Risto Laanoja and Vladimir Rogojin and Ahto Truu},
  title = {{Unicity}: Predicates and Atomic Swaps},
  year = {2026},
  howpublished = {{arXiv} preprint {arXiv}:2606.02192},
  url = {https://arxiv.org/abs/2606.02192}
}

@misc{polygonzkevm,
  author = {{Polygon zkEVM Team}},
  title = {{Polygon zkEVM}: A {ZK}-Rollup Compatible with {Ethereum}},
  year = {2023},
  howpublished = {Technical documentation},
  url = {https://docs.polygon.technology/zkEVM/}
}

@misc{babybear,
  author = {{Polygon Zero Team}},
  title = {{BabyBear}: A 31-bit Prime Field for {SNARK}s},
  year = {2024},
  publisher = {{GitHub}},
  journal = {{GitHub} repository},
  howpublished = {\url{https://github.com/Plonky3/Plonky3/tree/main/baby-bear}}
}

\end{document}